\documentclass[12pt, draftclsnofoot, onecolumn ]{IEEEtran}

\usepackage[dvips]{color}
\usepackage{amsthm}
\usepackage{epsf}
\usepackage{times}
\usepackage{epsfig}
\usepackage{graphicx}
\usepackage{pstricks}
\usepackage{amsmath}
\usepackage{amssymb}
\usepackage{amsxtra}
\usepackage{algorithm}
\usepackage{algpseudocode}
\usepackage{here}
\usepackage{rawfonts}
\usepackage{times}
\usepackage{url}
\usepackage{cite}
\usepackage{comment}
\usepackage[utf8]{inputenc}
\usepackage{caption}
\usepackage{subcaption}

\newtheorem{theorem}{\bf Theorem}

\newtheorem{proposition}{\bf Proposition}
\newtheorem{lemma}{\bf Lemma}

\newtheorem{definition}{\bf Definition}

\newcommand*\dif{\mathop{}\!\mathrm{d}}
\newcommand\norm[1]{\left\lVert#1\right\rVert}
\makeatletter

\IEEEoverridecommandlockouts
 
\ifCLASSINFOpdf
\else
\fi
  
\begin{document}	
\vspace{-0.3cm}
\title{\LARGE Predictive Deployment of UAV Base Stations in Wireless Networks: Machine Learning Meets Contract Theory
	\footnote{A preliminary version of this work appears in the proceedings of IEEE GLOBECOM 2018 \cite{zhang2018machine}. This research was supported by the U.S. National Science Foundation under Grant IIS-1633363.}}
	\vspace{-0.4cm}	
	
\author{\IEEEauthorblockN{ Qianqian Zhang$^1$, Walid Saad$^1$,  Mehdi Bennis$^2$, Xing Lu$^3$, M\'erouane Debbah$^{4,5}$, and Wangda Zuo$^{3}$}\vspace{-0cm}	 
	
	\IEEEauthorblockA{\footnotesize	 $^1$Bradley Department of Electrical and Computer Engineering, Virginia Tech, VA, USA,
		Emails: \url{{qqz93,walids}@vt.edu}\\
		$^2$Center for Wireless Communications-CWC, University of Oulu, Finland, Email: \url{mehdi.bennis@oulu.fi}\\
		$^3$Department of Civil, Environmental and Architectural Engineering, University of Colorado Boulder, CO, USA, Email:  \url{ {xing.lu-1,wangda.zuo}@colorado.edu} \\
		$^4$Mathematical and Algorithmic Sciences Lab, Huawei France R\&D, Paris, France, Email: \url{merouane.debbah@huawei.com}\\
		$^5$Large Systems and Networks Group (LANEAS), CentraleSupélec, Université Paris-Saclay,  Gif-sur-Yvette, France  
		\vspace{-1.4cm}		
	}
}

\maketitle

\begin{abstract}   
	In this paper, a novel framework is proposed to enable a predictive deployment of unmanned aerial vehicles (UAVs) as temporary base stations (BSs) to complement ground cellular systems in face of downlink traffic overload. First, a novel learning approach, based on the weighted expectation maximization (WEM) algorithm, is proposed to estimate the user distribution and the downlink traffic demand. Next, to guarantee a truthful information exchange between the BS and UAVs, using the framework of contract theory, an offload contract is developed, and the sufficient and necessary conditions for having a feasible contract are analytically derived. Subsequently, an optimization problem is formulated to deploy an optimal UAV onto the hotspot area in a way that the utility of the overloaded BS is maximized. Simulation results show that the proposed WEM approach yields a prediction error of around $10\%$. Compared with the expectation maximization and k-mean approaches, the WEM method shows a significant advantage on the prediction accuracy, as the traffic load in the cellular system becomes spatially uneven. 
	Furthermore, compared with two event-driven deployment schemes based on the closest-distance and maximal-energy metrics, the proposed predictive approach enables UAV operators to provide efficient communication service for hotspot users in terms of the downlink capacity, energy consumption and service delay. Simulation results also show that the proposed method significantly improves the revenues of both the BS and UAV networks, compared with two baseline schemes.   
		  
\end{abstract}   
 
\IEEEpeerreviewmaketitle

\section{introduction}

The use of unmanned aerial vehicles (UAVs) as flying base stations (BSs) has attracted growing interest in the past few years \cite{mozaffari2017mobile, zhang2018machine, bor2016efficient, zhang2018fast, khawaja2018survey, mozaffari2018tutorial,mozaffari2018beyond,saad2019vision}. 
UAVs can be deployed to complement the existing cellular systems, by providing reliable wireless services for ground users, to potentially increase the network capacity, eliminate coverage holes,  and cope with the steep surge of communication needs during hotspot events \cite{zhang2018machine}.  
Compared with the terrestrial BSs that are deployed at a fixed location for a long term, UAVs are more suitable for temporary on-demand service  \cite{bor2016efficient}. 
For instance, UAVs can provide communication service for major events (e.g. sport or musical events) during which the terrestrial network capacity is often strained \cite{zhang2018fast}. 
Furthermore,  UAVs can adjust their positions and establish line-of-sight (LOS) communication links towards ground users, thus improving network performance \cite{khawaja2018survey}. 
Due to their broad range of application domains and low cost, UAVs  is a promising solution to provide temporary connectivity for ground users\cite{mozaffari2018tutorial}.    

However, the UAV deployment for on-demand cellular service faces several key challenges.  
For instance,  UAVs are strictly constrained by their on-board energy, which should be efficiently used for communication.  
However, the on-demand deployment requires UAVs to continuously change their positions to meet instant communication requests. Therefore, most of on-board energy can be consumed by mobility, thus limiting their communication capabilities  \cite{zhang2018machine}.  
Moreover, to effectively alleviate network congestion during a hotspot event, the deployed UAV must have enough on-board power to satisfy the downlink communication demand. 
To allocate a qualified UAV with sufficient energy, the network operator should estimate the required transmit power, based on the real-time  traffic load.  
These challenges, in turn, motivate the need for a comprehensive prediction of cellular traffic, and a predictive approach for UAV deployment \cite{chen2019artificial}. 
To this end, machine learning (ML) techniques can be applied to  estimate  the cellular traffic demand within the target system. 
Given the predicted traffic load, each BS can detect hotspot areas and request suitable UAVs to alleviate network congestion.  

Another  challenge of the on-demand deployment for aerial wireless service is to incentivize cooperation between the ground BS and the UAV operators under the asymmetric information.    
As shown in \cite{hu2018uav}, the ground BSs and UAVs can belong to different operators who seek to selfishly maximize their individual benefits.  
Hence, to request a UAV's assistance, a ground BS must offer an appropriate economic reward to the UAV operator for aerial wireless service.   
However, given that the BS  has no prior knowledge of each UAV, there is no guarantee that the requested UAV is able to provide enough transmit power to satisfy the downlink demand.  
Therefore, designing an incentive mechanism is necessary to ensure a truthful information exchange between the UAV and BS systems, when the information among different network operators is asymmetric.   

\subsection{Related Works}

The optimal deployment of UAVs for cellular service has been studied in \cite{mozaffari2016optimal,  kalantari2016number, lyu2018uav}. 
In \cite{mozaffari2016optimal}, the authors studied the optimal locations and coverage areas of UAVs that minimizes the transmit power.  
The work in \cite{kalantari2016number} derived the minimum number of UAVs needed to satisfy the coverage and capacity constraints. 
In \cite{lyu2018uav},  the authors jointly optimized the  UAV trajectory and the network resource allocation  to maximize the throughput to ground users. 
The problem of traffic offloading from an existing wireless network to UAVs  has been addressed in \cite{sharma2016uav, lyu2017spectrum, cheng2018uav, sharafeddine2019demand}. 
In \cite{sharma2016uav}, the allocation problem of UAVs to each geographic area was investigated to improve the spectral efficiency and reduce the delay. 
In \cite{lyu2017spectrum} and \cite{ cheng2018uav}, the authors optimized the trajectory of UAVs  to provide wireless services to the cell-edge users.  
In \cite{sharafeddine2019demand},  an unsupervised learning approach was presented to solve the deployment of a fleet of UAVs for traffic offloading. 
However, most of the existing works \cite{mozaffari2016optimal, kalantari2016number, lyu2018uav, sharma2016uav, lyu2017spectrum, cheng2018uav, sharafeddine2019demand} assumed that the traffic demand of the cellular users is known a priori, which is challenging to estimate in a practical network.  
Furthermore, the works \cite{mozaffari2016optimal, kalantari2016number, lyu2018uav, sharma2016uav, lyu2017spectrum, cheng2018uav, sharafeddine2019demand} optimized the performance of the cellular network in a centralized approach which assumes all UAVs belong to the same entity.  
Given the fact that the UAVs can belong to multiple operators, a new framework is needed to consider the individual utility of UAVs in the aerial communication service, while optimizing the performance of the ground cellular networks.   

Meanwhile, in \cite{li2017learning, yu2017modeling,  valente2017survey}, a number of ML approaches are proposed to predict the traffic demands of cellular networks.  
In \cite{li2017learning}, a prediction framework is proposed to model the cellular data in the temporal and spatial domains.  
The authors in \cite{yu2017modeling} predicted the locations of users during daily activities, based on pattern modeling. 
The work \cite{valente2017survey} provided surveys that focused on the general use of ML algorithms in cellular networks.  
Furthermore, the prior art in \cite{chen2018liquid, chen2017learning, amorim2017machine} studied the use of ML techniques to improve the performance of UAV-aided  communications.  
In \cite{chen2018liquid}, an ML framework based on liquid state machine is proposed to optimize the caching content and resource allocation for each UAV.  
In \cite{chen2017learning}, the authors investigated an ML approach to construct a radio map  for  autonomous path planning of UAVs.  
In \cite{amorim2017machine}, ML algorithms are applied to detect aerial users from the ground mobile users.   
However, most of the works in \cite{li2017learning,yu2017modeling, valente2017survey ,chen2018liquid, chen2017learning,amorim2017machine} aim to build an ML model to predict regular traffic patterns, while hotspot  events are considered as an anomaly and excluded from  these studies.  
In fact, none of the approaches proposed in  \cite{li2017learning,yu2017modeling, valente2017survey ,chen2018liquid, chen2017learning,amorim2017machine} can  effectively identify the hotspot areas or accurately predict excessive traffic load  during the hotspot event.  
Thus,  results of these prior works cannot enable a predictive UAV deployment  for on-demand cellular service to alleviate the traffic congestion. 

\subsection{Contributions}
 
The main contribution of this paper is a novel framework for optimally deploying UAVs to assist a ground cellular network in alleviating its downlink traffic congestion during hotspot events.  
The proposed framework divides the deployment process  into four, inter-related and sequential stages:  learning stage,  association stage,  movement stage, and service stage. 
For each stage, we evaluate the performance of the proposed framework, using an open-source dataset in \cite{cityCellularTrafficMap}. 
Our main contributions include:
\begin{itemize}
\item 
A novel framework, based on the weighted expectation maximization (WEM) approach, is proposed to predict the downlink traffic demand for each cellular system in the learning stage.  
The proposed WEM method is a general version of the conventional expectation maximization (EM) algorithm, which enables a variable weight at each data point in the  distribution modeling. 
In particular, the proposed approach identifies the user distribution, predicts the cellular data demand, and  pinpoints the hotspot areas within the cellular system.    

\item 
In the association stage,  to employ a UAV with sufficient on-board energy to satisfy the downlink demand, 
the framework of contract theory \cite{bolton2005contract} is introduced, where each overloaded BS can jointly design the transmit power and unit reward of the target UAV. 
We analytically derive the sufficient and necessary conditions needed to guarantee a truthful information exchange between the BS and UAV operators. 
The proposed contract approach yields little communication overhead and exhibits a low computational  complexity.  

\item   
Simulation results show that the mean relative error (MRE) of the proposed ML approach is around $10\%$.  
Compared with two baselines, an EM scheme and a $k$-mean algorithm, the proposed method yields a better prediction accuracy, particularly when the  downlink traffic load in the cellular system becomes spatially uneven.  
Furthermore, simulation results show that the designed contract ensures a non-negative payoff of each UAV, and each UAV will truthfully reveal its communication capability by accepting the contract designed for itself.   

\item  
We evaluate the performance of the proposed approach with two event-driven allocation methods, based on the closest-distance and maximal-energy metrics, that deploy a target UAV after the network congestion occurs, without traffic prediction and contract design.  
Numerical results show that the proposed predictive method enables UAV operators to provide efficient downlink service for hotspot users, in terms of  the downlink capacity, energy consumption, and service delay. 
Moreover,  the proposed method significantly improves the economic revenues of both the BS and UAV networks, compared with two   baseline schemes.   
\end{itemize}

The rest of this paper is organized as follows.  
In Section \ref{sec:sm}, we present the system model. The problem formulation is given in Section \ref{sec:pf}. 
In Section \ref{sec:ml}, the ML approach is proposed to predict downlink traffic demands. 
In Section \ref{contractDesign},  the feasible contract is designed with the optimal UAV being employed to offload the cellular traffic.   
Simulation results are presented in Section \ref{sec:simulation}. Finally, conclusions are drawn in Section \ref{conclusion}.

\section{System Model}\label{sec:sm}

Consider a set $\mathcal{I}$ of $I$ cellular BSs providing downlink wireless service to a group of user equipments (UEs) in a geographical area $\mathcal{A}$. 
Each BS $i \in \mathcal{I}$ serves an area  $\mathcal{A}_i$, such that $\cup_{ \forall i \in \mathcal{I}}   \mathcal{A}_i = \mathcal{A}$, and $\mathcal{A}_i \cap \mathcal{A}_k = \emptyset$ for any  $i\ne k \in \mathcal{I}$. 
The spatial distribution of the served UEs for each BS $i$ is denoted by $f_i(\boldsymbol{y})$, where $\int_{\boldsymbol{y} \in \mathcal{A}_i} f_i(\boldsymbol{y}) \dif \boldsymbol{y} =1$.   
A set $\mathcal{J}$ of $J$ flying UAVs can provide additional cellular service,  if the hotspot events happen in the ground cellular network.  
We assume that the group BSs and UAVs belong to different network operators, and different frequency bands  are used for the ground and aerial downlink transmissions, separately. 
A single antenna is equipped at each UE that can receive signals from both the ground BS and the UAV.     
Initially, a UE will connect to one of the ground BSs.  
However, as shown in Fig. \ref{systemmodel},  if a ground BS $i \in \mathcal{I}$ is overloaded in the downlink, BS $i$ can request the assistant of a UAV to offload the service of some UEs.  
We assume that  a UAV only serves the UEs of a single BS at each time, while each BS can employ multiple UAVs, based on the cellular traffic demand.  
In this regard, if the downlink traffic demand at the level of a given BS is excessive, such that no single UAV is capable to alleviate traffic congestion, then the BS will divide the offloaded UEs into multiple spatially-disjoint sets, and request an individual UAV for each  UE set, independently. 
Meanwhile, each UAV is equipped with a directional antenna array that enables beamforming transmissions \cite{zhu20193d}.  
As a result, interference between different UAV networks is negligible.     

\begin{figure}[!t]
	\begin{center}
		\vspace{-3 cm}
		\includegraphics[width=12cm]{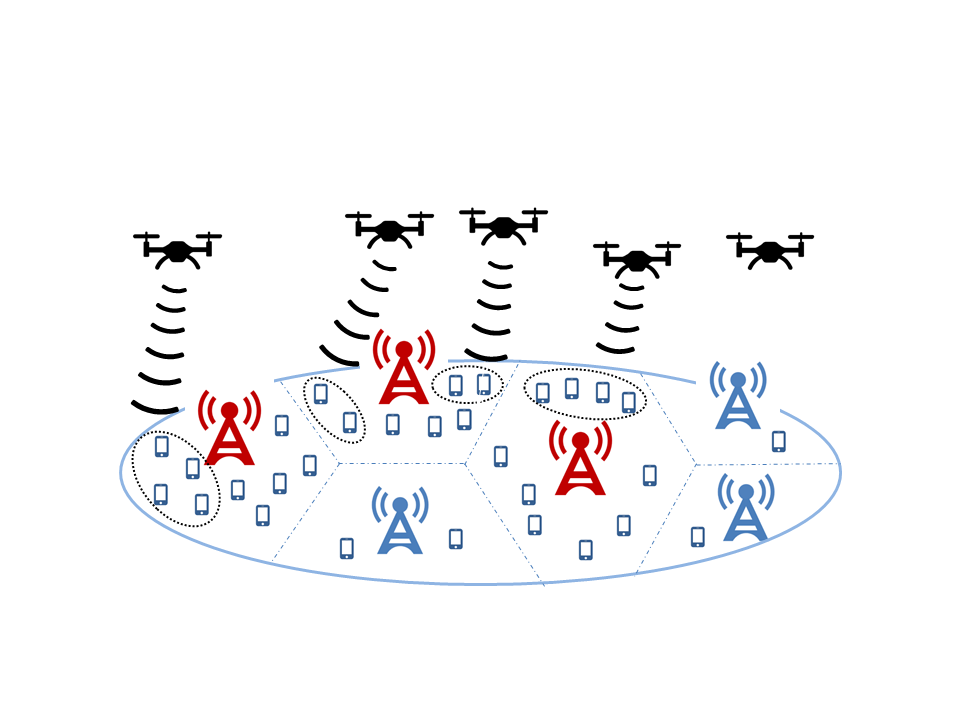}
		\vspace{-1.5cm}
		\caption{\label{systemmodel}\small The red BSs are having excessive traffic load in the downlink, thus each red BS requests a UAV to offload a part of UEs to the aerial cellular system. }
	\end{center}\vspace{-1.1cm}
\end{figure}

\subsection{Air-to-ground downlink communications}
The path loss of the air-to-ground communication link from a typical UAV located at $\boldsymbol{x} \in \mathbb{R}^3$ to a typical ground UE that is located at $\boldsymbol{y} \in \mathbb{R}^3 $ can be given by \cite{al2014modeling}:
\begin{align}\label{channelGain}
h[dB](\boldsymbol{x},\boldsymbol{y}) =  20 \log \left( \frac{4 \pi f_{c} \norm{\boldsymbol{x}-\boldsymbol{y}}}{c} \right) + \xi(\boldsymbol{x},\boldsymbol{y}), 
\end{align}
where  $f_{c}$ is the carrier frequency of UAV downlink communications, $ \norm{\boldsymbol{x}-\boldsymbol{y}}$ is the UAV-UE distance, $c$ is the speed of light, 
and $\xi(\boldsymbol{x},\boldsymbol{y})$ is the additional path loss  of the air-to-ground channel, compared with the free space propagation.  
The value of $\xi(\boldsymbol{x},\boldsymbol{y})$  can be modeled as a Gaussian distribution with different parameters $(\mu_{\text{LOS}},\sigma^2_{\text{LOS}})$ and $(\mu_{\text{NLOS}},\sigma^2_{\text{NLOS}})$  for the LOS and non-line-of-sight (NLOS) links,  respectively. 
Then, the achievable data rate from a UAV $j \in \mathcal{J}$ located at $\boldsymbol{x}_j$ to a UE located at $\boldsymbol{y} \in \mathcal{A}_i$ is 
\begin{align}\label{dataRate}\vspace{-0.2cm}
	r_{ij}(\boldsymbol{x}_j,\boldsymbol{y},p_j) =  w  \log_2 \left( 1+  \frac{ g(\boldsymbol{x}_j,\boldsymbol{y})  p_{j} }{h(\boldsymbol{x}_j,\boldsymbol{y}) w n_0}\right),  
\end{align}
where $w$ is the downlink bandwidth of each UAV,  $g(\boldsymbol{x}_j,\boldsymbol{y}) $ is the antenna gain of  UAV $j$ towards the UE located at $\boldsymbol{y}$,  $p_{j}$ is the transmit power of UAV $j$, $h(\boldsymbol{x}_j,\boldsymbol{y}) $ is the path loss in linear scale, 
and $n_0$ is the average noise power spectrum density at the UE.  
The probability of having a LOS link between UAV $j$ located at $\boldsymbol{x}_j$ and the UE located at $\boldsymbol{y}$ is given by \cite{al2014optimal}: 
\begin{equation}
P_{\text{LOS}}(\boldsymbol{x}_j,\boldsymbol{y}) = \frac{1}{1 + a \exp(- b[ \frac{180}{\pi}\varphi(\boldsymbol{x}_j,\boldsymbol{y}) -a] ) },
\end{equation} 
where $a$ and $b$ are constant values that depend on the communication environment, 
$\varphi(\boldsymbol{x}_j,\boldsymbol{y}) = \sin^{-1}(\frac{H_{j}}{\|\boldsymbol{x}_j- \boldsymbol{y} \|})$ is the elevation angle, and $H_{j}$ is the altitude of UAV $j$.  
Consequently, the average downlink rate between a UAV $j$ and a UE at $\boldsymbol{y} \in \mathcal{A}_i$ will be: 
\begin{equation}
	\bar{r}_{ij}(\boldsymbol{x}_j,\boldsymbol{y},p_j) = P_{\text{LOS}}(\boldsymbol{x}_j,\boldsymbol{y}) \cdot r^{\text{LOS}}_{ij}(\boldsymbol{x}_j,\boldsymbol{y},p_j) + (1-P_{\text{LOS}}(\boldsymbol{x}_j,\boldsymbol{y})) \cdot r^{\text{NLOS}}_{ij}(\boldsymbol{x}_j,\boldsymbol{y},p_j).  
\end{equation}
 
In order to serve multiple downlink UEs,  each UAV applies a time-division-multiple-access (TDMA) technique\footnote{The focus of this work is on the deployment stage and, hence, we do not optimize the multiple access scheme type or operation. Optimizing multiple access can be done post-deployment and will be subject to future work.} 
that divides the time resource evenly among all served UEs, and all bandwidth will be allocated to one single UE during each time slot \cite{lyu2016cyclical}.  
By using suitable uplink control signals, the UAV-UE channel can be accurately measured, and thus, the beamforming of UAV's antennas can be properly optimized towards the served UE.  
Consequently, the average rate that UAV $j$ can provide to the hotspot UEs from BS $i$ will be
\begin{equation} 
\begin{aligned}\label{networkCapacity}
C_{ij}(\boldsymbol{x}_j,p_j) & =  \int_{\mathcal{A}^{c}_i}  \bar{r}_{ij}(\boldsymbol{x}_j,\boldsymbol{y},p_j) f^c_i(\boldsymbol{y}) \dif \boldsymbol{y},  
\end{aligned}
\end{equation}
where  $\mathcal{A}^{c}_i \subset  \mathcal{A}_i$ is the hotspot area,  
$f^c_i(\boldsymbol{y})$ is the normalized spatial distribution of UEs within  $\mathcal{A}_i^c$, and $\int_{\mathcal{A}^{c}_i} f^c_i(\boldsymbol{y}) \dif \boldsymbol{y}=1$.  
When downlink congestion occurs, BS $i$ detects the congested area $\mathcal{A}^{c}_i$ and offloads the UEs within $\mathcal{A}^{c}_i$ to the target UAV.

\begin{table}[t!]  \vspace{-0.8cm}
	\centering
	\caption{Summary of our notations}
	\label{table1} \scriptsize
	\begin{tabular}{ |p{0.9 cm}|p{6.6cm}||p{0.9cm}|p{6.6cm}|  } 
		\hline
		 Notation  &    Description  &  Notation   &   Description  \\
		\hline 
		$I$, $J$  & Number of BSs and number of UAVs  &  $t_{ij}$  & Movement time of UAV $j$ to the service location of BS $i$     \\ 
		\hline
		$T$ & Interval of the UAV's offloading service  & $\bar{r}_{ij}$ &  Average rate of UAV $j$ to each hotspot user of BS $i$     \\
		\hline
		$\boldsymbol{y}$    & Location of a ground user  &   $C_{ij}$  & Average rate  of UAV $j$ to all hotspot users of BS $i$     \\
		\hline
		$\boldsymbol{x}_j$, $\boldsymbol{x}_{ij}^{*}$  & Current location  of UAV $j$, and service location of UAV $j$ associated with BS $i$  &  $B_{ij}$& Amount of data that UAV $j$ provides to all hotspot users of BS $i$ within one $T$  \\
		\hline
		$f_i,S_i$ & User distribution and data demand distribution of BS $i$ &   $\rho_i$, $\rho_i^c$  &  Average rate demand  per user/ hotspot user of BS $i$   \\ 
		\hline
		$\mathcal{A}_i$, $\mathcal{A}_i^c$  & Service area and hotspot area of BS $i$    &  $U_{ij}$ & Utility of BS $i$ by employing UAV $j$ \\
		\hline
		$Q_i$, $Q_i^c$ & Number of all users and number of hotspot users of BS $i$   & $R_{ij}$  & Utility of UAV $j$ by  providing offloading service to BS $i$  \\
		\hline
		$d_i$  & Data demand of hotspot users within one $T$ of BS $i$ & $\theta_{ij}$ & Type of UAV $j$  with respect to BS $i$ \\
		\hline
		$p_{j}$  &Transmit power of UAV $j$   &  $\boldsymbol\omega$,  $\boldsymbol{\pi}$    & Weight vectors in the user and demand distribution models  \\
		\hline
		$u_i$    & Unit payment of BS $i$ & $\boldsymbol{\mu}$, $\boldsymbol{\Sigma}$&      Mean  and covariance  of Gaussian distribution  \\  
		\hline 
	\end{tabular}  
\end{table}

\subsection{UAV deployment process}

\begin{figure}[!t]
	\begin{center}
		\vspace{-0.2cm}
		\includegraphics[width=12cm]{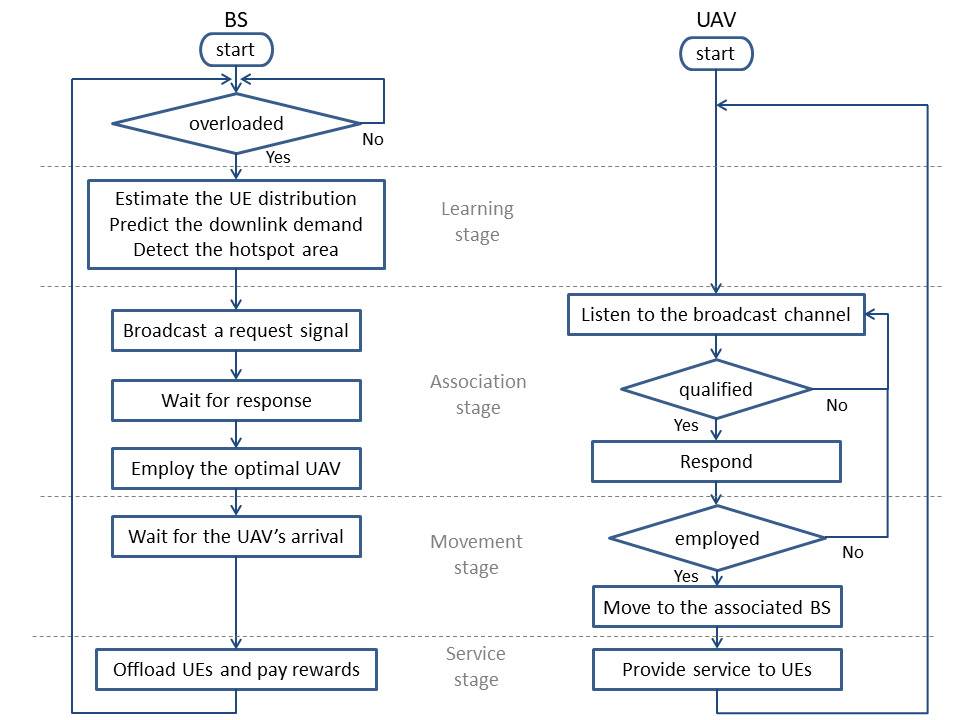}
		\caption{\label{deploymentprocess}\small  Flowchart of the proposed UAV predictive deployment process for each BS (left) and each UAV (right).  }  
	\end{center}\vspace{-1cm}
\end{figure}

Given the average downlink rate of each UAV in (\ref{networkCapacity}),  the next step is to deploy suitable UAVs to offload the traffic  and alleviate the downlink congestion in the ground cellular network. 
To facilitate the analysis, we assume that the  service interval of each UAV a constant $T$.  
As shown in Fig. \ref{deploymentprocess}, the deployment process has four sequential stages:  
learning stage, association stage, movement stage, and service stage. 
The details of each stage are given as next: 

\subsubsection{Learning stage}

For each BS $i \in \mathcal{I}$, once the downlink traffic exceeds its network capacity, a learning stage with a fixed duration $\tau$ starts.  
During $\tau$, BS $i$ collects the  transmission record  $\mathcal{S}_i=\{ (s, \boldsymbol{y},t) | \boldsymbol{y} \in \mathcal{A}_i, t \in [\Delta t,2\Delta t,\cdots,\tau] \}$,    
where $s$ is the  data rate that BS $i$ provides to the UE located at $\boldsymbol{y}$ at time $t$, and $\Delta t$ is the time slot during which the downlink rate can be considered to be constant. 
Given that the hotspot area $\mathcal{A}_i^c$ and the UE distribution $f_i(\boldsymbol{y})$ is unknown, 
a learning stage is necessary for BS $i$ to estimate the spatial distribution of UEs and the traffic demand of the on-going hotspot event. 
Considering common events, such as sport games and outdoor concerts, where mobile users are often confined to seat or geographically constrained spaces, the mobility of hotspot UEs is scarce. 
Thus, we assume that the UE distribution $f_i(\boldsymbol{y})$ during one $T$ is time-invariant.  
Furthermore, to estimate the traffic demand within the congested area, a spatial density function ${{S}_i}(\boldsymbol{y})$ is proposed to evaluate the average data rate per UE at each location $\boldsymbol{y} \in \mathcal{A}_i$.  
The proposed approach for estimating the UE distribution and traffic demand will be discussed in Section \ref{sec:ml}.  
Consequently, the total data demand $d_i$ from a hotspot area $\mathcal{A}^{c}_i$ during a time interval $T$ will be given by: 
\begin{equation}\label{dataDemand}\vspace{-0.2cm}
d_i =  \int_t^{t+T}  \int_{\boldsymbol{y} \in \mathcal{A}^{c}_i} {S}_i(\boldsymbol{y}) \dif \boldsymbol{y} \dif t =  T  \int_{\boldsymbol{y} \in \mathcal{A}^{c}_i} {S}_i(\boldsymbol{y}) \dif \boldsymbol{y}.
\end{equation}

Next, the  BS will estimate the necessary number of UAVs to alleviate downlink congestion and calculate the optimal service location of each target UAV.   
Following from \cite[equation~(42)]{mozaffari2017mobile} and \cite[equations~(10)~and~(11)]{mozaffari2016optimal}, given the UE distribution $f^c_i(\boldsymbol{y})$ and the hotspot area  $\mathcal{A}^{c}_i$, the optimal  location $\boldsymbol{x}^*_{ij}$ of a target UAV $j$ in serving BS $i$ can be derived in a way to minimize the transmit power $p_{ij}(\boldsymbol{x}_{ij}^{*},\rho^c_i)$, while satisfying the average rate requirement $\rho^c_i$ per UE. 
The average rate per UE is defined  by the ratio of the sum data rate within the hotspot area $\mathcal{A}_i^c$ over the total number of hotspot UEs $Q^c_i$, where $\rho^c_i= \frac{d_i}{T Q_i^c} $. 
Thus, the optimal service location of the target UAV can be calculated by BS $i$, prior to the UAV's deployment. 
We define $p_{\text{max}}$ to be the maximum transmit power of each UAV, which is limited by the antennas' hardware, and  $\eta \in (0,1)$ to be the ratio of efficient transmission time to the service time $T$, due to the signal overhead and the channel measurement process. 
If $d_i > \eta T C_{ij}(\boldsymbol{x}_{ij}^{*}, p_\text{max})$, then even though a UAV is located at the optimal service point $\boldsymbol{x}_{ij}^{*}$ and it applies the maximum  transmit power $p_\text{max}$, the downlink demand $d_i$ cannot be satisfied. 
In this case, using a single UAV $j \in \mathcal{J}$ is no longer sufficient to offload the hotspot traffic. 
Therefore,  BS $i$ will evenly divide the hotspot area $\mathcal{A}_i^c$, based on the downlink data demand, into $N$ disjoint areas $\{ \mathcal{A}_i^c(n)\}_{n = 1,\cdots,N}$,  where  $\int_{\boldsymbol{y} \in \mathcal{A}^c_i(n) } S_i(\boldsymbol{y}) \dif \boldsymbol{y} = \frac{d_i}{N}$,  
and $N$ is the smallest integer needed to guarantee that, for each subset $n= 1,\cdots,N$,  the following requirement holds: 
\begin{equation}\label{multiUAV} \vspace{-0.2cm}
	d_i(n) = \frac{d_i}{N}< \eta T  C_{ij}(\boldsymbol{x}_{ij}^{*}(n),p_{\text{max}}). 
\end{equation}    
For each $n = 1,\cdots, N$, BS $i$ will deploy a UAV onto the service point $\boldsymbol{x}_{ij}^{*}(n)$ to offload the downlink traffic with the subarea $\mathcal{A}_i(n)$.  
The requests of multiple UAVs to different subareas are sequential and independent at each round $n = 1, \cdots, N$.

\subsubsection{Association stage  } 
In the association stage,  each overloaded BS $i$ requests the assistance of a UAV, by broadcasting a signal with the downlink demand $d_i(n)$ and the service location $\boldsymbol{x}_{ij}^{*}(n)$ for each subset $n$.  
A first-call-first-serve scheme is applied, and each BS $i \in \mathcal{I}$ will listen to the broadcast channel before sending the signal.  
If the channel is occupied by another BS, then BS $i$ will wait until the on-going association  is completed. 
For each BS $i$, the  goal is to request a UAV that has enough on-board power to meet the downlink demand $d_i $ of UEs within $\mathcal{A}_i^c$. 
The optimal UAV association to each overloaded BS will be studied in Section \ref{contractDesign}.  

\subsubsection{Movement stage}
After the association stage, the selected UAV $j$ starts to move from its current location $\boldsymbol{x}_j$ to the service point $\boldsymbol{x}_{ij}^{*}$ of its target BS $i$. The duration $t_{ij}$ of the movement stage depends on the distance $\|\boldsymbol{x}_j - \boldsymbol{x}_{ij}^{*} \|$ and the average speed $v_j$ of UAV $j$. 

\subsubsection{Service stage}
Once it reaches the service point,  UAV $j$ will provide downlink communications to its group of associated UEs for a  time period $T - t_{ij}$.  
Note that, during the movement and service stages, the employed UAV is fully dedicated to its associated BS. Thus, the UAV cannot be requested by any other BSs until the end of its current service.  
Furthermore, to guarantee a sufficient service time, the maximum travel time of UAV $j$ is limited by $t_{ij}\le \kappa_iT$, where $\kappa_i \in (0,1)$. If the travel time exceeds $\kappa_iT$, UAV $j$ is not a potential choice for BS $i$.

After the service stage ends, the BS-UAV association will end. Then, UAV $j$ will listen to the broadcast channel, if its remaining on-board energy $E_j$ can support another service period  $T$; otherwise, the UAV will move to a nearby recharging station. 
We assume that a number of recharging stations are deployed, such that a UAV can access a recharging station within a short flight time from any location in $\mathcal{A}$. Thus, the movement energy to a recharging station is negligible to effect the BS-UAV association results. 
In order to optimally associate UAVs to each overloaded BS, we first define a utility function that each BS aims to maximize when selecting a UAV to offload cellular traffic in Section \ref{subsec:UBS}. 
Next, the UAV's utility function  is given in Section \ref{subsec:UUAV} that defines its economic payoff  from serving a ground BS. 

\subsection{Utility function of a ground BS} \label{subsec:UBS} 
In  TDMA downlink transmissions, the employed UAV $j$ evenly divides the service time $T-t_{ij}$ to each hotspot  UE. 
Therefore,  based on the average downlink rate in (\ref{networkCapacity}),
the achievable data amount that UAV $j$ can provide to the UEs of BS $i$ is 
\begin{equation}
\begin{aligned}\label{aerialCapacity}
B_{ij}(p_{j}) = \eta (T-t_{ij}) C_{ij}(p_{j}). 
\end{aligned}
\end{equation} 
Note that, the movement duration $t_{ij}$ and the transmit power $p_j$ are private information for UAV $j$, and, thus, BS $i$ cannot know their values during the service request process.   
Then, the utility of BS $i$, by employing UAV $j$ to offload the excess cellular traffic, will be:   
\begin{equation}
	\begin{aligned}\label{U_BS}
	U_{ij}(u_i,p_j,d_i ) &= \beta  B_{ij}(p_j) - u_i d_i, 
	\end{aligned}
\end{equation} 
where $\beta $ is the payment from UEs to BS $i$ (per bit of downlink data), and  $u_i$ is the unit payment that BS $i$ gives to UAV $j$ (per bit of aerial data service). 
Thus, the first term in (\ref{U_BS}) represents the reward that BS $i$ gets from its UEs by employing UAV $j$ to provide aerial cellular service, and
the second term is the total payment that BS $i$ gives to UAV $j$. 

\subsection{Energy model and utility function of a UAV}\label{subsec:UUAV} 

In the considered problem, the power consumption of  each UAV consists of three main components: the transmit power $p_j$, the propulsion power $m$, and the hovering power $p_h$.   
For tractability and as done in \cite{zeng2019energy}, we ignore the acceleration and deceleration stages during the UAV's movement, and the propulsion power $m$ is considered as a constant for a fixed flying speed.  
Then, the travel time $t_{ij}$ can be uniquely determined based on the moving distance $\| \boldsymbol{x}_j - \boldsymbol{x}_{ij}^* \|$. 
During the service stage, the maximum available  power that UAV $j$ can use for downlink transmissions  will be  $p_{ij}^{\text{max}} = \frac{E_j-  m t_{ij} - p_h (T-t_{ij})}{T-t_{ij}}$,  where $m t_{ij}$ is the energy consumed during the UAV's movement, and $p_h (T-t_{ij}) $ is the hovering energy during the service stage.  
Therefore, we have the transmit power $p_{j} \in [p_{ij}(\boldsymbol{x}_{ij}^{*},\rho^c_i),  \min \{p_{ij}^{\text{max}},p_{\text{max}}  \} ]  $, where $p_{ij}(\boldsymbol{x}_{ij}^{*},\rho^c_i)$ is the  minimum required power to satisfy the downlink data demand, and  $p_{\text{max}}$ is the maximum transmit power.  
Without loss of generality, we assume that $p_{ij}(\boldsymbol{x}_{ij}^{*},\rho^c_i) \le \min \{p_{ij}^{\text{max}},p_{\text{max}}  \}$ holds. Otherwise, UAV $j$ is not a potential option for BS $i$.  
Consequently, the utility that a UAV $j \in \mathcal{J}$ can achieve from providing the aerial cellular service to the UEs of BS $i$ will be: 
\begin{align}\label{U_UAV}
R_{ij}(u_i,p_{j},{d}_i) = u_i {d}_i - \alpha [ p_{j}(T-t_{ij})+ p_h(T-t_{ij}) + m t_{ij}],
\end{align}
where $\alpha$ is a unit cost per Joule of UAV's on-board energy.  
The first term in (\ref{U_UAV}) is the reward that UAV $j$ obtains from  BS $i$, and the second term  is the energy cost. 

\section{Problem formulation}\label{sec:pf}

The objective of an overloaded BS is to employ a suitable UAV with sufficient on-board power to offload excessive cellular traffic, while maximizing the utility function in (\ref{U_BS}). 
Meanwhile, the goal of each UAV is to optimize its utility in (\ref{U_UAV}). 
However, by comparing (\ref{U_BS}) and  (\ref{U_UAV}), we realize that  $\arg \max_{u_i,p_j} U_{ij} =\arg \min_{u_i,p_j} R_{ij}$ and $\arg \max_{u_i,p_j} R_{ij} =\arg \min_{u_i,p_j} U_{ij}$. Therefore, each BS-UAV pair has conflicting interests. 
Given that the BSs and UAVs belong to different operators, each will maximize its own utility. The conflict between each BS and each UAV is irreconcilable. 

Meanwhile, since the values of the unit payment $u_i$ and the data demand $d_i$ will be broadcast by BS $i$ during the association stage, each UAV $j$ has all necessary information to determine its utility. 
However, BS $i$ cannot easily acquire some private information of each UAV, such as its current location and onboard energy, which causes the asymmetric information. 
Since private information of each UAV determines its travel time to a BS and the downlink communication capacity, it is essential for the BS to have accurate information to evaluate the service performance of each UAV.   
In order to guarantee a truthful information exchange, 
each BS $i$ can jointly design $(u_i,p_j)$ to ensure  mutual benefit for both the BS and UAV operators, so that the conflict of interest can be properly resolved. 
Therefore, we let $\phi_{ij} = (u_i,p_{j})$ be a \emph{traffic offload contract}, which captures the values of $p_j$ and $u_i$ if BS $i$ employs UAV $j$ to offload its hotspot UEs.  
In order to understand the relationship between the unit payment $u_i$  and the transmit power $p_j$, we divide both sides of (\ref{U_UAV}) by $\alpha(T-t_{ij})$ and rewrite the utility of UAV $j$ as follows:  
\begin{equation}
	\begin{aligned}\label{newU_UAV}
	\tilde{R}_{ij}(u_i,p_{j},d_i) &= \frac{ {d}_i }{\alpha (T-t_{ij})} u_i - p_j - \frac{ mt_{ij} }{T-t_{ij}} - p_h,\\ 
	&= \theta_{ij} u_i -  p_{j} -M_{ij}, 
	\end{aligned}
\end{equation}
where the values of $\theta_{ij} = \frac{ {d}_i }{\alpha (T-t_{ij})}$ and $M_{ij} = \frac{ mt_{ij} }{ T-t_{ij}} + p_h$ are determined for each BS-UAV pair.

Since $\theta_{ij}$ determines the sensitivity of $\tilde{R}_{ij}$ to the increase of  $u_i$ and  $p_{j}$ in (\ref{newU_UAV}), its value is essential for the joint design of $(u_i,p_j)$.  
Therefore, we define $\theta_{ij}$   as the \emph{type} of UAV $j$ with respect to BS $i$, where    $\theta_{ij} \in \Theta_i = [\frac{d_i}{\alpha T}    ,   \frac{d_i}{\alpha (1-\kappa_i)T}   ]$.  
Note that,  due to the privacy of  $t_{ij}$,  the type $\theta_{ij}$ of each UAV $j \in \mathcal{J}$ is unknown for BS $i$.  
In order to design the contract without knowing each UAV's type, before broadcasting the request signal,  BS $i$ will design a set of contracts 
$\Phi_i(\Theta_i) = \{\phi_{ij}(\theta_{ij})  | \forall \theta_{ij} \} =  \{(u_i(\theta_{ij}),p_{j}(\theta_{ij}) )| \forall \theta_{ij} \}$  for all UAV types $\theta_{ij} \in \Theta_i$,      
where $u_i(\theta_{ij})$ represents the payment that BS $i$ pays to UAV $j$ per bit of data, given that UAV $j$ is of type $\theta_{ij}$, and $p_{j}(\theta_{ij})$ is the transmit power that UAV $j$ of type  $\theta_{ij}$ provides  to serve  BS $i$.  
Then, (\ref{newU_UAV}) becomes $\tilde{R}(\theta_{ij}) = \theta_{ij} u_i(\theta_{ij}) - p_j(\theta_{ij}) - M_{ij}$.  
Meanwhile, to ensure that a UAV will accept the contract of its own type, two constraints,   based on contract theory \cite{bolton2005contract}, must be considered, which are individual rationality (IR) condition and incentive compatibility (IC) condition. 
 
\begin{definition}
	[Individual Rationality] A contract designed by BS $i$ satisfies the IR constraint, if a UAV of any type $\theta_{ij} \in \Theta_i$ will receive a non-negative payoff from BS $i$ by accepting the contract item for type $\theta_{ij}$, i.e.  
	$\theta_{ij}  u_i(	\theta_{ij} ) -  p_{j}(	\theta_{ij} ) - M_{ij} \ge 0$, $\forall \theta_{ij} \in \Theta_i$.  
\end{definition}  
A contract satisfying the IR condition guarantees that the reward that each UAV $j  \in \mathcal{J}$ can obtain from serving BS $i$ is great than or equal to  zero. 
Compared with the non-employed state in which the payoff is always zero, each UAV is willing to accept the contract from the requesting BS, as long as its contract satisfies the IR condition.      
\begin{definition}
	[Incentive Compatibility] A contract designed by BS $i$ satisfies the IC constraint, if a UAV of type $\theta_{ij}$ will get the highest utility from BS $i$ by accepting the contract designed for its own type $\theta_{ij}$, compared with all the other types $\theta$ in $\Theta_i$, i.e. 
	$\theta_{ij} u_i(\theta_{ij}) -  p_{j}(\theta_{ij}) - M_{ij} \ge \theta_{ij}  u_i(\theta ) -  p_{j}(\theta ) - M_{ij}$, $ \forall   \theta  \in \Theta_i$. 
\end{definition} 
A contract satisfying IC condition guarantees that each UAV $j$ will only accept the contract designed for its own type $\theta_{ij}$, since accepting the contract of any other type $\theta \in \Theta_i$ will  result in a lower or the same reward. 
A contract satisfying both IR and IC conditions is called a \emph{feasible} contract, which ensures the UAV will accept and only accept the contract designed for its  type. 


Consequently, for each overloaded BS $i \in \mathcal{I}$, the objective is to maximize its utility in (\ref{U_BS}), by estimating the downlink data demand ${d}_i$ within the hotspot area $\mathcal{A}^{c}_i$, designing the contract set $\Phi_i$ for each UAV of any type in $\Theta_i$, and determining an optimal UAV $j \in \mathcal{J}$ to offload the excessive cellular service.  
We formulate this predictive UAV deployment problem as follows,  
\begin{subequations}\label{problem_BS}   
	\begin{align}  
	\max_{\substack{    \{(u_i(\theta_{ij}), p_{j}(\theta_{ij}) )| \forall \theta_{ij} \} , j \in \mathcal{J}}} \quad &  U_{ij} (u_i(\theta_{ij}),p_{j}(\theta_{ij}), {d}_i ), \label{obj1} \\
	\textrm{s. t.} \quad  
	& R_{ij}(\theta_{ij}) \ge 0,    \label{con_IR}\\  
	& R_{ij}(\theta_{ij}) \ge R_{ij}(\theta), \forall \theta  \in \Theta_i  \label{con_IC}, \\ 
	& p_{ij}(\boldsymbol{x}_{ij}^{*},\rho^c_i) \le p_j(\theta_{ij}) \le \min \{p_{ij}^{\text{max}},p_{\text{max}}  \}, \label{con_power}\\  
	& t_{ij}  \le \kappa_i T, \label{con_time}\\  
	& {d}_i > 0,  u_i(\theta_{ij}) > 0. \label{con_other}  
 	\end{align}
\end{subequations} 
The objective function (\ref{obj1}) is the utility that BS $i$ obtains from employing UAV $j$ of type $\theta_{ij}$.   
(\ref{con_IR}) and (\ref{con_IC}) are the IR and IC constraints, respectively. 
(\ref{con_power}) is the constraint on the transmit power, and (\ref{con_time}) limits the maximum travel time.  
(\ref{con_other}) imposes a positive downlink demand within $\mathcal{A}_i^c$, and a positive unit payment. 
Here, (\ref{con_IC}) itself is an optimization problem, which must be first addressed to satisfy the IC condition. 
Since the selection of $\theta_{ij}$ will jointly determine the values of the objective function and all constraints in (\ref{problem_BS}),  $\theta_{ij}$ becomes the key variable to find the optimal association result.  
To simplify the optimization problem (\ref{problem_BS}), we first derive the necessary and sufficient conditions for IC and IR constraints, based on the UAV type $\theta_{ij}$, which essentially reduces to the problem of designing a feasible contract. 
Consequently,  to solve the predictive UAV deployment problem in (\ref{problem_BS}), 
first, a learning-based approach  is proposed to predict the downlink demand ${d}_i$  in Section \ref{sec:ml}. Next, the  traffic offload contract  $\Phi_{i}$ is developed in Section \ref{contractDesign}, with the optimal UAV being selected to maximize the utility of BS $i$.   

\section{Learning Stage: Estimation of Cellular Traffic Demand }\label{sec:ml}
 
In this section, our goal is to estimate the UE distribution and the downlink data demand  during a hotspot event. 
This estimation is necessary to solve (\ref{problem_BS}) because the data demand $d_i$ is needed  to determine the type $\theta_{ij}$ of each UAV $j$ with respect to  BS $i$.   
To enable an accurate modeling,  BS $i$ collects the downlink transmission records during the learning stage.   
For notation simplicity, let $N$ be the total number of records, 
and $\mathcal{S}_i$ can be rewritten as $ \{(s_n, \boldsymbol{y}_n, t_n) | n = 1, \cdots, N \}$.  
In Section \ref{UEdistribution}, we extract the spatial distribution $f_i(\boldsymbol{y})$ of the downlink UEs, and then, in Section \ref{demandfunction} the downlink data rate $S_i(\boldsymbol{y})$ is modeled and the hotspot area $\mathcal{A}_i^c$ is determined. 
Consequently, the downlink data demand $d_i$ can be given by   (\ref{dataDemand}). 

\subsection{Estimation of the UE distribution} \label{UEdistribution} 
Given $\mathcal{S}_i$,  BS $i$ can model the UE distribution, using the location information  $\mathcal{Y}$ = $\{ \boldsymbol{y}_1,  \cdots,\boldsymbol{y}_N\}$. 
We assume that each UE's location follows a latent distribution $f_i(\boldsymbol{y})$, and each  $ \boldsymbol{y}_n$ is an independent sample from this distribution. 
A Gaussian mixture model (GMM), which is the weighted sum of multiple Gaussian distributions, can model the UE's distribution, as follows: 
\begin{equation}\label{GMM}
	f_i(\boldsymbol{y}) = \sum_{l=1}^{L} \omega_l \mathcal{N}(\boldsymbol{y}| \boldsymbol{\mu}_l, \boldsymbol{\Sigma}_l),  
\end{equation} 
where $L$ is the number of Gaussian distributions, and
$\omega_l \in (0,1)$, $\boldsymbol{\mu}_l$, and $\boldsymbol{\Sigma}_l$ are the weight, mean and variance of the $l$-th Gaussian, respectively, with $\sum_{l} \omega_l = 1$. 
The value of $\omega_l$ represents the probability that the data point $\boldsymbol{y}$ is generated by the $l$-th distribution. 
GMM has been widely applied in  \cite{kasgari2018human,selim2015modeling,christopher2016pattern} to model the distribution of a latent variable based the sampled data.  
Due to its special feature of multiple clusters, GMM is particularly appropriate to model the UE distribution in the congested  area,  where each hotspot area  corresponds to a Gaussian center. 

Given the  location record $\mathcal{Y}$,  the expectation-maximization (EM) algorithm \cite{christopher2016pattern} is applied to optimize the parameters $\{\omega_l, \boldsymbol{\mu}_l, \boldsymbol{\Sigma}_l \}_{ l=1,\cdots,L}$   in (\ref{GMM}) via an iterative approach, which maximizes a log-likelihood function 
$	\ln p (\mathcal{Y}| \boldsymbol{\omega},\boldsymbol{\mu},\boldsymbol{\Sigma}) = \ln \Pi_{n=1}^N \left(  \sum_{l=1}^{L} \omega_l \mathcal{N} (\boldsymbol{y}_n|\boldsymbol{\mu}_l,\boldsymbol{\Sigma}_l) \right)$.   
After initialization, the EM algorithm alternates between the E and M steps. 
First, in the E step, the posterior probability that  $\boldsymbol{y}_n$ is generated by the $l$-th Gaussian is calculated by 
\begin{equation}\label{responsibility}
	v_{nl} = \frac{\omega_l \mathcal{N} (\boldsymbol{y}_n|\boldsymbol{\mu}_l,\boldsymbol{\Sigma}_l)}{\sum_{z=1}^{L} \omega_z \mathcal{N}(\boldsymbol{y}_n|\boldsymbol{\mu}_z,\boldsymbol{\Sigma}_z)}. 
\end{equation} 
Then, in the M step, the parameters are updated using the posterior probability (\ref{responsibility})  by 
\begin{equation}
	\boldsymbol{\mu}_l = \frac{\sum_n v_{nl}\boldsymbol{y}_n }{ \sum_n v_{nl}}, ~~~ \boldsymbol{\Sigma_l} =  \frac{\sum_n v_{nl}(\boldsymbol{y}_n - \boldsymbol{\mu}_l) (\boldsymbol{y}_n - \boldsymbol{\mu}_l)^T}{ \sum_n v_{nl}}, ~~~ \omega_l = \frac{\sum_n v_{nl}}{ N}. 
\end{equation}
After each EM iteration, the updated parameters will result in an increase of the log-likelihood function, and the algorithm is guaranteed to converge to a local optimum  \cite{christopher2016pattern}.  

\subsection{Estimation of the downlink data rate} \label{demandfunction}

In order to predict the downlink demand $d_i$, each BS $i$ needs to capture the spatial feature of the cellular traffic. 
Based on  the assumption of the time-invariant data demand, 
we define the traffic \textit{density} $\bar{{S}_i}(\boldsymbol{y})$ at each location $\boldsymbol{y} \in \mathcal{A}_i$  as  the time-average downlink rate at $\boldsymbol{y}$ during the learning stage, where 
$	\bar{{S}_i}(\boldsymbol{y})  =  \frac{1 }{\tau}\sum_{(s_n,\boldsymbol{y},t_n) \in \mathcal{S}_i} s_n \Delta t$.  
In order to generate  a continuous model $\bar{S}_i(\boldsymbol{y})$  that captures the spatial features of the downlink traffic density, 
a Gaussian mixture function (GMF) is proposed  as follow,  
\begin{equation}\label{GMF}
	{S}_i(\boldsymbol{y})  = \sum_{k=1}^{K} \pi_k \exp \left(  \frac{-(\boldsymbol{y}-\boldsymbol{\mu}_k)^T \boldsymbol{\Sigma}_k^{-1} (\boldsymbol{y}-\boldsymbol{\mu}_k)}{2} \right) ,
\end{equation}    
where $K$ is the number of basis functions, and $\pi_k$, $\boldsymbol{\mu}_k$, and $\boldsymbol{\Sigma}_k$  are the coefficient,  mean and variance of the $k$-th Gaussian function. 
Thus, the traffic density at location $\boldsymbol{y}$ is modeled by the sum of $K$ Gaussian functions with coefficient $\{ \pi_k \}_{k=1,\cdots,K}$. 

Note that, the GMF in (\ref{GMF}) is different from the GMM in (\ref{GMM}).   
First,  a GMM has a probabilistic interpretation, while a GMF is a deterministic function that calculates the traffic density at each location $\boldsymbol{y}$ by adding the values of $K$ Gaussian functions with different coefficients.  
Second,  the sum of each coefficient $\pi_k$ in GMF represents the total volume of downlink traffic demand. Thus, it is always greater than one, which make a difference from the unit weight-sum in GMM.  

To properly model the downlink traffic density  $\bar{S}_i$,  the parameters $\{\pi_k,\boldsymbol{\mu}_k,\boldsymbol{\Sigma}_k \}_{k=1,\cdots,K}$ in (\ref{GMF}) need to be optimized. 
Since the EM method associates the same weight to all data points, it is not suitable to the traffic density modeling, because each data point $\boldsymbol{y}_n$ can have a different traffic density $\bar{{S}_i} (\boldsymbol{y}_n)$.  
In order to adapt the weight of each location $\boldsymbol{y}_n$ in determining the parameter values  according to the traffic density $\bar{{S}_i} (\boldsymbol{y}_n)$, as well as to capture the spatial diversity of the traffic load within the cellular network,    
a weighted expectation maximization (WEM) algorithm  is proposed to optimize the parameters in the traffic density model ${ {S}_i}(\boldsymbol{y})$. 

In the proposed WEM method, the initial value of each Gaussian center $\boldsymbol{\mu}_k$ is the location $\boldsymbol{y}_k$ that has the $k$-th highest traffic density in $\bar{{S}_i}(\boldsymbol{y})$. 
The initial variance $\boldsymbol{\Sigma}_k$ equals the identity matrix with the equal weight $\pi_k = \frac{1}{K}\sum_{\boldsymbol{y}}  \bar{{S}_i}(\boldsymbol{y}) $.  
Then, the WEM algorithm updates  $\{\pi_k,\boldsymbol{\mu}_k,\boldsymbol{\Sigma}_k \}_{k=1,\cdots,K}$ via an iterative approach.  
In the E step,  the percentage that each Gaussian function $k$ contributes to  the traffic density at location $\boldsymbol{y}_n$ is evaluated via  
$	v_{nk} = \frac{ \pi_k \mathcal{N}(\boldsymbol{y}_n | \boldsymbol{\mu}_k,\boldsymbol{\Sigma_k})}{\sum_{k=1}^K \pi_k \mathcal{N}(\boldsymbol{y}_n | \boldsymbol{\mu}_k,\boldsymbol{\Sigma_k})}$. 
Next, in the M step, the parameters of each Gaussian function will be updated in a weighted approach, where the mean $\boldsymbol{\mu}_k$ is recalculated via 
\begin{equation}
	\boldsymbol{\mu}_k = \frac{\sum_n  v_{nk}  \boldsymbol{y}_n   \bar{{S}_i} (\boldsymbol{y}_n)}{\sum_n v_{nk} \bar{{S}_i}(\boldsymbol{y}_n)},  
\end{equation}
which is a  sum of all locations $\boldsymbol{y}_n \in \mathcal{Y}$, weighted by the posterior probability $v_{nk}$ and the traffic density $\bar{{S}_i}(\boldsymbol{y}_n)$. Thus, a location $\boldsymbol{y}_n$ with a higher traffic density  $\bar{{S}_i} (\boldsymbol{y}_n)$ will have a higher weight in determining the value of $\boldsymbol{\mu}_k$, and the center of Gaussian $k$ will gradually be driven closer to the high-density locations.  
Similarly, the variance $\boldsymbol{\Sigma}_k$ and the linear coefficient $\pi_k$ of each Gaussian function is also updated, with weights $\bar{{S}_i}(\boldsymbol{y}_n)$, by 
\begin{equation}
	\boldsymbol{\Sigma}_k = \frac{\sum_n v_{nk} (\boldsymbol{y}_n - \boldsymbol{\mu}_k) (\boldsymbol{y}_n - \boldsymbol{\mu}_k)^T  \bar{{S}_i}(\boldsymbol{y}_n)  }{\sum_n v_{nk} \bar{{S}_i}(\boldsymbol{y}_n)}, ~~~ \pi_k = \frac{\sum_n v_{nk}\bar{{S}_i}(\boldsymbol{y}_n)}{\sum_k\sum_n v_{nk}\bar{{S}_i}(\boldsymbol{y}_n)}. 
\end{equation}  
Furthermore, similar to the EM approach, a WEM method will converge to a local optimum, which maximizes the weighted conditional log-likelihood function  \cite{zhang2018machine}. 

Although the EM and WEM methods have similar mathematical expressions, the physical meaning and iterative process are fundamentally different. 
First, different from the  unit sum-weight in (\ref{GMM}),  the weight-sum of the WEM method  represents the total volume of the downlink data demand, which can be any positive value. 
Second, when updating the Gaussian parameters, the proposed WEM method considers the traffic density at each location and assigns a  higher weight to the location with higher demand in the density model. 
In contrast, EM method associates each point with an equal weight. Thus, 
the spatial diversity of the traffic load cannot be properly captured. 
Therefore,  the proposed WEM approach expands the application range of the EM scheme, and can be seen as a general version of EM, which models the distribution with a variable weight at each data point.

The hotspot area $\mathcal{A}_i^c$ is a location set in which the  traffic destiny is much higher than other locations in $\mathcal{A}_i$.   
Given the traffic density model $S_i$,  the average traffic density in $\mathcal{A}_i$ is given by $ \bar{s_i} = \frac{1}{| \mathcal{A}_i  |} \int_{\boldsymbol{y} \in \mathcal{A}_i} S_i(\boldsymbol{y} )  \dif \boldsymbol{y}$, where $| \mathcal{A}_i  |$ denotes the area of $\mathcal{A}_i$. 
Then, 
by calculating the traffic density at each Gaussian center $\{ \boldsymbol{\mu}_k\}_{k=1,\cdots,K}$, 
the mean $\boldsymbol{\mu}_k^*$ with the highest traffic density is chosen, and its neighborhood area, where the traffic density is higher than $ \bar{s_i}$ forms the hotspot area $\mathcal{A}_i^c$.  
The downlink UEs within  $\mathcal{A}_i^c$ will be offloaded to the aerial cellular network. 
Based on  the traffic density model $S_i(\boldsymbol{y})$ and the hotspot area $\mathcal{A}_i^c$, the predicted data amount $d_i$ for a time interval $T$ can be calculated based on (\ref{dataDemand}). 

Given the downlink traffic demand $d_i$ and the UE distribution $f_i(\boldsymbol{y})$, all variables in (\ref{problem_BS}) have determined values, except for the unit payment $u_i$ and the transmit power $p_j$. 
Next, in order to to solve (\ref{problem_BS}), we will jointly decide the value of $(u_i,p_j)$, by designing the feasible contract between an overloaded BS $i$ with each UAV $j \in \mathcal{J}$.

\section{Association stage: Contract Design and UAV Allocation}\label{contractDesign}

\subsection{Contract design }\label{subsct:fc}   
Given the predicted  traffic demand ${d}_i$,  a BS $i \in \mathcal{I}$  can request  UAVs to offload the UEs within the hotspot area $\mathcal{A}_i^c$, so that the future downlink congestion can be alleviated.    
However, to employ a qualified UAV  to meet the downlink demand, each BS needs to carefully design the contract $\Phi_{i} = \{ (u_i(\theta_{ij}),p_{j}(\theta_{ij})) | \forall \theta_{ij} \in \Theta_i \}$ for UAVs of any type $\theta_{ij}$.   
The feasible contract satisfying the IR and IC conditions can guarantee 
that each UAV $j \in \mathcal{J}$ will accept the contract designed for its own type and provide the required downlink transmissions. 
To develop a feasible contract set, we first analyze the sufficient and necessary conditions for a feasible contract. 

\begin{proposition} \label{pro}
	[Necessary Condition]  For any $\theta_{ij},\theta_{ij}^{'} \in \Theta_i $, if $\theta_{ij}  > \theta_{ij}^{'}$, then $u_i(\theta_{ij}) \ge u_i(\theta^{'}_{ij})$ and $p_{j}(\theta_{ij} ) \ge p_{j}(\theta^{'}_{ij})$. 
\end{proposition}
\begin{proof}
See Appendix \ref{AppPro1}. 
\end{proof}
Proposition 1 shows that for a typical UAV $j$,  if its type with respect to a typical BS $i$ increases from $\theta_{ij}^{'}$ to  $\theta_{ij}$, then it will receive a higher unit payment $u_i(\theta_{ij}) \ge u_i(\theta_{ij}^{'})$, and in return, it should provide a larger transmit power $p_j(\theta_{ij}) \ge p_j(\theta_{ij}^{'})$. 
Given that  $\theta_{ij} = \frac{d_i}{\alpha (T-t_{ij}) }$, a higher type $\theta_{ij}$ indicates either a higher downlink demand $d_i$, or a longer travel time $t_{ij}$. In the first case, if the downlink demand is higher, the employed UAV must increase the transmit power  to satisfy the larger traffic needs. Thus, $p_j(\theta_{ij})$ will increase. 
On the other hand, if UAV $j$ travels for a long time $t_{ij}$, it consumes more  energy on movement, which requires a higher unit payment $u_i(\theta_{ij})$ to compensate for the energy cost. 
Therefore, a  UAV of a higher type is required to provide more transmit power, and will be given a higher unit payment.  
The conclusion in Proposition \ref{pro} will lead to the necessary and sufficient conditions of a feasible contract, as shown next. 

\begin{theorem}\label{contracttheorem} 
	A contract set $\Phi_i = \{ (u_i(\theta_{ij}),p_{j}(\theta_{ij}))|\forall \theta_{ij} \}$ satisfies IR and IC  
	constraints, if and only if all the following three conditions hold:  
	(a) $\frac{\dif p_{j} (\theta_{ij}) }{\dif \theta_{ij}} \ge 0$ and $\frac{\dif u_i (\theta_{ij})}{\dif \theta_{ij}} \ge 0$,   
	(b) $\theta^{\text{min}}u_i(\theta^{\text{min}}) - p_{j}({\theta^{\text{min}}}) - M_{ij} \ge 0$,  
	(c) $\frac{\dif p_{j}(\theta_{ij})}{\dif \theta_{ij} }  = \theta_{ij} \cdot \frac{\dif u_i(\theta_{ij}) }{\dif \theta_{ij} }$.  
\end{theorem}
\begin{proof}
	See Appendix \ref{AppTheo1}.
\end{proof}  
Theorem \ref{contracttheorem} gives the necessary and sufficient conditions for a contract set $\Phi_i$ to jointly satisfy the IC constraint in (\ref{con_IC}) and the IR constraint in (\ref{con_IR}).    
Therefore, each feasible solution of Theorem \ref{contracttheorem} can guarantee that a UAV only accepts the contract designed for its own type, and provides the required transmit power to meet the downlink demand.   
Here, we note that Theorem \ref{contracttheorem} results in a loose solution set. 
In essence, all of contracts from this solution set meet the necessary and sufficient conditions of the IC and IR requirements, 
and, thus, they are optimal in the contract-theoretic problem.  
Meanwhile, Theorem \ref{contracttheorem} 
provides each BS with more freedom to choose the feasible contract based on its real-time communication need. 
In order to minimize the communication overhead in the association stage, we aim to propose a contract with  the lowest complexity and the least  broadcast overhead.  
Therefore, to enable an efficient BS-UAV association, we propose the best contract with  $ \frac{\dif u_{i}(\theta_{ij})}{\dif \theta_{ij} }= \gamma_i >0$.   
Consequently, the feasible contract that is proposed by BS $i$ is given as follows.  
\begin{lemma}\label{lemma1}
	Under the condition that $ \frac{\dif u_{i}(\theta_{ij})}{\dif \theta_{ij} }= \gamma_i $, the feasible contract between  BS $i$  and a UAV $j$ of type $\theta_{ij}$ is $\phi_{ij} = (u_i,p_j) = (\gamma_i \theta_{ij}, \gamma_i \theta_{ij}^2/2) $, where $ \gamma_i = \frac{ 2\alpha^2 T^2 p_h}{d_i^2}$. 
\end{lemma}
\begin{proof}
	Based on $ \frac{\dif u_{i}(\theta_{ij})}{\dif \theta_{ij} }= \gamma_i $ and condition (c) of Theorem 1, we have $u_i = \gamma_i \theta_{ij}$, $p_j = \gamma_i \theta_{ij}^2/2$, and condition (a) holds naturally. 
	For BS $i$, the minimal UAV type  is $\theta^{\text{min}} = \frac{d_i}{\alpha T}$, when $t_{ij}=0$. 
	Therefore, condition (b) becomes $\gamma_i \ge \frac{2M_{ij}}{{\theta^{\text{min}}}^2 } = \frac{ 2\alpha^2 T^2 p_h}{d_i^2} $. 
	Therefore, we set $ \gamma_i = \frac{ 2 \alpha^2 T^2 p_h}{d_i^2} $. 
\end{proof} 
Therefore, for each overloaded  BS $i$, the designed contract is $ (u_i,p_j) = (\gamma_i \theta_{ij}, \gamma_i \theta_{ij}^2/2) $ with $ \gamma_i = \frac{ 2\alpha^2 T^2 p_h}{d_i^2} $ for each UAV in $\mathcal{J}$ with any type $\theta_{ij}$. 

\subsection{The optimal UAV association under the feasible contract}

Given the feasible contract set $ \{ (\gamma_i \theta_{ij}, \gamma_i \theta_{ij}^2/2) | \forall \theta_{ij} \}$, the utility $R_{ij}(\theta_{ij})$ of each candidate UAV $j \in \mathcal{J}$ and the utility $U_{ij}(\theta_{ij})$ of the requesting BS $i$ can be jointly determined.  
Then, the optimization problem in (\ref{problem_BS}) becomes 
\begin{subequations}\label{opt_new}
	\begin{align} 
	\max_{\substack{ j \in \mathcal{J}}} \quad &  U_{ij} ( \theta_{ij}  ), \label{newobj1} \\
	\textrm{s. t.} \quad    
	& p_{ij}(\boldsymbol{x}_{ij}^{*},\rho^c_i) \le p_j(\theta_{ij}) \le \min \{p_{ij}^{\text{max}},p_{\text{max}}  \}, \label{newcon_power}  \\
	& t_{ij} \le \kappa_i T  \label{newcon_time}. 
	\end{align}
\end{subequations} 
Therefore,  BS $i$ aims to find a UAV of the optimal type $\theta_{ij}^*$ that maximizes its utility in (\ref{newobj1}), while satisfying  (\ref{newcon_power}) and (\ref{newcon_time}).  
In the association stage, after BS $i$ sends the request  signal,  each UAV $j$ will respond with its type $\theta_{ij}$. 
Based on the derivation $\frac{\dif U_{ij}(\theta_{ij})}{\dif \theta_{ij}}<0$,  the optimal UAV is $j^{*} =  \arg \max_{j \in \mathcal{J}_i} U(\theta_{ij}) = \arg \min_{j \in \mathcal{J}_i} \theta_{ij} $, where $\mathcal{J}_i = \{ j|  p_{ij}(\boldsymbol{x}_{ij}^{*},\rho^c_i) \le \frac{\gamma_i}{2} \theta_{ij}^2  \le \min \{p_{ij}^{\text{max}},p_{\text{max}}\} , t_{ij} \le \kappa_i T \}$.    
Thus, the qualified UAV with a smallest type is the optimal solution. 
The complete process of the predictive UAV deployment is summarized in Algorithm \ref{algo_all}.   

\begin{algorithm}[t!]  
	\caption{Proposed process for the UAV predictive deployment } \label{algo_all} 
	\begin{algorithmic}\small 
		\State For each BS $i \in \mathcal{I}$, once downlink communication exceeds the network capacity, \textbf{do}:	 \\   
		\textbf{1. Learning stage:} \\
		\quad (a) BS $i$ collects  $\mathcal{S}_i$ to model the UE distribution $f_i(\boldsymbol{y})$, estimate  the downlink  traffic density  ${S}_i(\boldsymbol{y})$, and detect\\  
		\quad \quad  the hotspot area $\mathcal{A}_i^c$ based on the WEM approaches proposed in Section \ref{sec:ml}. \\
		\quad (b)  BS $i$ calculates the downlink demand $d_i$ of the offloaded UEs via (\ref{dataDemand}),  estimates the number $N$ of required \\
		\quad \quad  UAVs through (\ref{multiUAV}), and computes the service point $\boldsymbol{x}_{ij}^{*}$ for each target UAV $j$, based on the solution in \cite{mozaffari2016optimal}.\\  
		\textbf{2. Association stage:}  for $n = 1, \cdots, N$: \\
		\quad (a) BS $i$ listens to the broadcast channel. If the channel is occupied, wait; otherwise, BS $i$ broadcasts the \\
		\quad \quad request signal with $d_i(n),\boldsymbol{x}_{ij}^{*}(n)$, $\kappa_i$, and $\Phi_{i}(n) = \{ \gamma_i \theta_{ij},\frac{\gamma_i}{2} \theta_{ij}^2 | \forall \theta_{ij} \}$, where $\gamma_i = \frac{ 2 \alpha^2 T^2 p_h}{{d_i(n)}^2} $;\\ 
		\quad (b) Each UAV $j \in \mathcal{J}$ listens the broadcast channel. After receiving the request from BS $i$, each UAV calculates \\ 
		\quad \quad  the movement time $t_{ij}$, its UAV type $\theta_{ij}$ with respect to BS $i$, and the available transmit power $p_{ij}^{\text{max}}$ after\\
		\quad \quad  arriving at $\boldsymbol{x}_{ij}^{*}$. If $p_{ij}^{\text{max}}\ge \frac{\gamma_i}{2} \theta_{ij}^2$ and $t_{ij}\le \kappa_i T$, UAV $j$ replies  $\theta_{ij}$ to BS $i$; otherwise, ignore. \\  
		\quad (c)  BS $i$ identifies the feasible UAV set $\mathcal{J}_i$, and employs the optimal UAV $j^{*} = \arg \min_{j \in \mathcal{J}_i} \theta_{ij}$.\\  
		\quad (d)  If $n=N$, BS $i$ releases the broadcast channel; otherwise, go back to 2(a). \\
		\textbf{3. Movement stage:} The employed UAV $j^{*}$ starts to move towards the service point of the requesting BS $i$.  \\ 
		\textbf{4. Service stage:} \\
		\quad (a) BS $i$ pays $\gamma_i \theta_{ij^{*}} d_i$, and offloads the UEs within $\mathcal{A}_i^c$ to UAV $j^{*}$.  \\
		\quad (b) UAV $j^{*}$ provides the downlink service with a transmit power $p_{j^{*}} = \frac{\gamma_i}{2}\theta_{ij^{*}}^2$ for a service time $T - t_{ij^{*}}$. \\
		\textbf{End} 
	\end{algorithmic}
\end{algorithm}

Compared with conventional UAV deployment, the contract-based optimization has three advantages. 
First, the proposed method not only reveals that the closest UAV is the optimal solution, but it also optimally determines the amount of the payment that the BS should offer the UAV, such that the utility of the BS can be maximized and the utility of the UAV is non-negative.  
Second, based on IC constraint, each UAV will receive the highest utility by accepting the contract designed for its real type. 
Thus, the use of contract theory allows us to capture the economic incentive of each UAV, forcing it to truthfully tell the requesting BS with its actual type, which is unknown to the BS a priori. 
Therefore, the proposed contract approach guarantees a truthful information exchange between the BS and UAV operators,  which the traditional optimization method cannot achieve. 
In the end, the proposed algorithm is more efficient for practical implementation, due to less information exchange in the association stage.  
In contrast, conventional optimization techniques will require all necessary information from all UAVs for solving the centralized association problem. 
Hence, compared with traditional optimization methods, our proposed algorithm reduces the  communications overhead, exhibits a  lower communication overhead, 
and  ensure a truthful information exchange between the BS and UAV operators.

\section{Simulation results and Analysis}\label{sec:simulation}

\subsection{Simulation parameters}
For our simulations, we consider a UAV-assisted wireless network in a dense urban environment, operating at the $2$ GHz frequency with a downlink bandwidth of $20$ MHz. 
The parameters in the LOS probability model are $a = 9.6$ and $b = 0.28$ \cite{al2014optimal}. 
The Gaussian parameters of the additional air-to-ground path loss are $\mu_{\text{LOS}}=1.6$ and $\sigma_{\text{LOS}} = 8.41$ for the LOS link while $\mu_{\text{NLOS}}=23$ and $ \sigma_{\text{NLOS}} =33.78$  for the NLOS case \cite{al2014modeling}.   
For the UAV parameters, based on the specifications in  \cite{DJI3},  
we set the mobility power $m=20$ W with an average moving speed of $5$ m/s, and  the hovering power is $p_h=16$ W.  
The maximal on-board energy of each UAV is $25$ Wh, and the battery recharge takes $10$ minutes.    
The maximum downlink transmit power  is $p_{\text{max}}=20$ W, and the unit cost for on-board energy is  $\alpha = 1.2$.  
For each UE, the  noise power spectral density is $-174$ dBm/Hz, and the data service per bit  is $\beta =10^{-7}$.   
For the UAV deployment process,  we set  $\Delta T=1$ second,  the learning duration $\tau=2$ minutes, and the service time $T = 18$ minutes.  
The ratio of efficient transmission in each time slot is $\eta = 90\%$, and the maximum ratio of the UAV's movement duration over the time interval is $\kappa_i = 0.1$. 

\subsection{Dataset description and preprocessing} 

\begin{figure*}[!t]
	\begin{center}
		\vspace{-1cm}
		\begin{subfigure}{.49\textwidth}
			\centering
			\includegraphics[width=8cm]{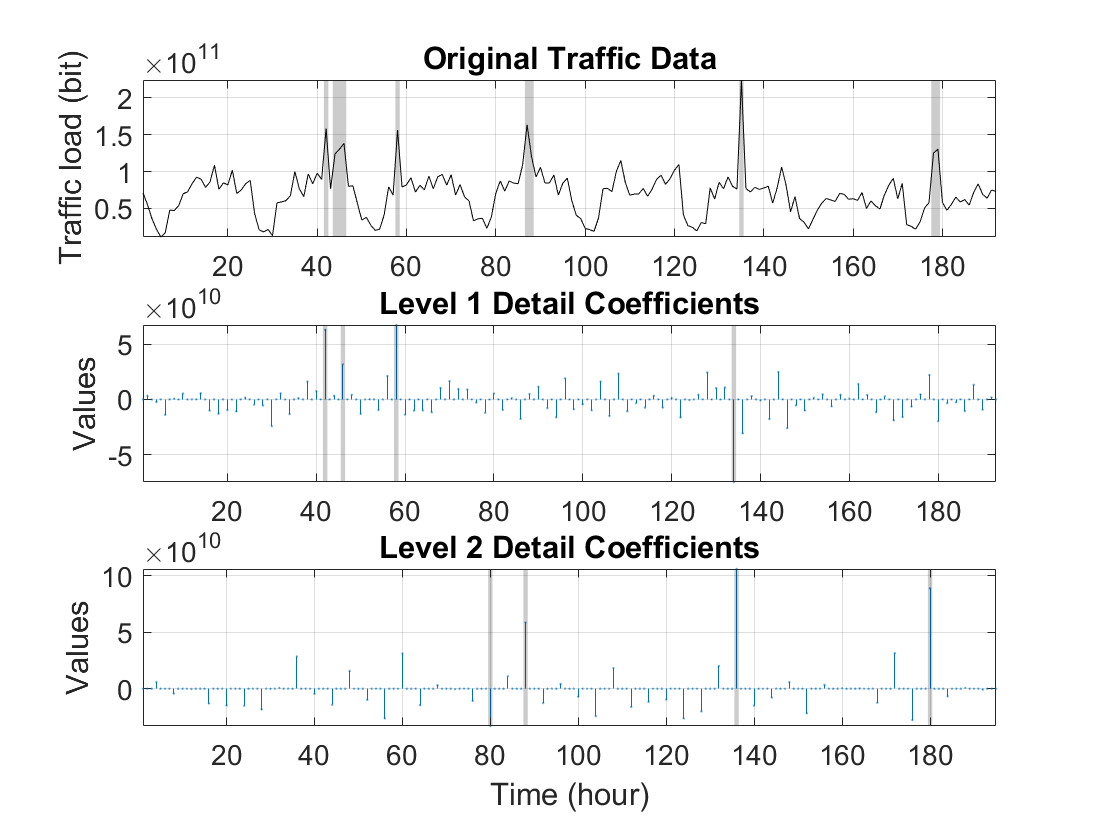}\vspace{0 cm}
			\caption{\label{ML:dwt} Two-level DWT components. }
		\end{subfigure}
		\begin{subfigure}{.49\textwidth}
			\centering
			\includegraphics[width=8cm]{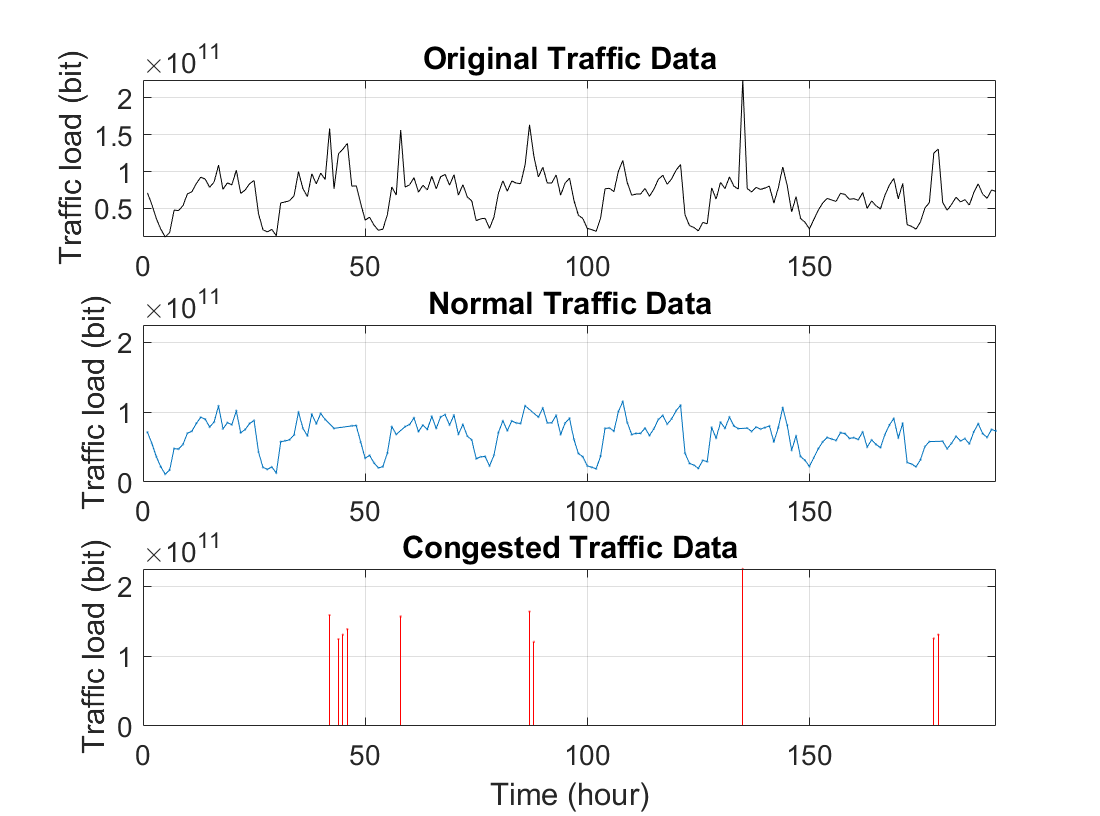}\vspace{0 cm}
			\caption{\label{ML:wtResult} Normal traffic states and potential congestion events. }
		\end{subfigure}
		\vspace{-0.3 cm} 
		\caption{\small{\label{anamolydetection} Two-level DWT is applied to detect the cellular traffic congestion from a city level. 
			} 
		}
	\end{center}
	\vspace{-1.2 cm}
\end{figure*}

An open-source dataset ``city-cellular-traffic-map" in \cite{cityCellularTrafficMap} is used for the modeling, training, and testing of the proposed UAV deployment framework. 
The dataset collects HTTP traffic data through the cellular networks during each hour within a middle-sized city of China from August 19 to August 26, 2012. The dataset consist of two parts. 
One lists the identification number (ID) and the location in longitude and latitude of each BS, and the other collects the number of UEs, packets and traffic data that each BS transmits to downlink UEs during each hour. 
In order to identify  hotspot events in the dataset, we apply the discrete wavelet transform (DWT) to the hourly cellular traffic in the city level. 
As shown in the upper figure of Fig. \ref{anamolydetection}, the  cellular traffic within the city area presents a conspicuously periodic pattern, with several sudden and erratic surges. 
DWT processes the time-serial data by analyzing both the value and frequency components,  where the lower-frequency component defines the long term trend, and the higher-frequency component represents the small-scale rapid variation. 
A hotspot event usually causes a steep surge in the traffic amount. Therefore, such rapid change can be captured by DWT in the higher frequency domain. 
As shown in Fig. \ref{ML:dwt}, a two-level DWT is applied to detect the frequency change of cellular traffic, and the gray bars mark the time points when the traffic amount has a sudden increase.  
Based on the result, the dataset is separated into the normal traffic data  and the potential congested traffic, as given in Fig. \ref{ML:wtResult}. 
Here, we find a time window from $42$ to $47$, which is 18 to 23 p.m. on August 20, that shows a continuously high cellular traffic amount, and the hotspot event is highly likely to happen during this period. Therefore, the traffic data from 42 to 47 are used for  the predictive UAV deployment in the following analysis. 

However, the data in \cite{cityCellularTrafficMap} does not include the location information of each UE, or the service area of each BS.  
To identify the UE distribution and the  traffic density, the location and time labels are generated and attached to each  transmission record  via the following  approach.  
First, the service area $\mathcal{A}_i$ of each BS $i$ is partitioned, based on the closest-distance principle.   
Next, we use the total packet number to denote the number of downlink transmissions. 
Furthermore, we note that the original time label $t$ in \cite{cityCellularTrafficMap} is based on one hour, which is too coarse to enable our analysis.  
To extract the estimated data with a desired duration, a new label with a finer time grain of one second is randomly generated and attached to each traffic record.  
Then, given $\tau= 2$ and $T = 18$, we divide each hour evenly into three intervals, such that the cellular data during first two minutes of each interval is used to model the UE distribution and downlink traffic, and data from the following $18$ minutes is used to estimate the UAV's transmission performance. 
Eventually, the location label $\boldsymbol{y}_n$ of each traffic record is generated by a GMM with random parameters to which we add a zero-mean Gaussian noise with a standard deviation of three meters. 
With additional location and time labels, the dataset is suitable for the studied problem. 

\subsection{Performance of the cellular traffic prediction }

\begin{figure*}[!t]
	\begin{center}
		\vspace{-1cm}
		\begin{subfigure}{.49\textwidth}
			\centering
			\includegraphics[width=8cm]{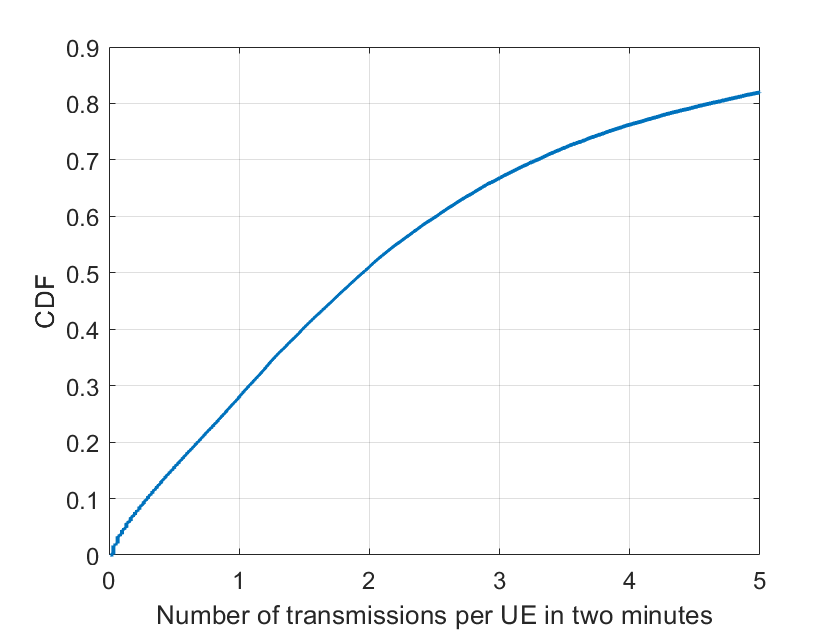}\vspace{0 cm}
			\caption{\label{sim:transPerUE2min} Number of transmissions per UE per two minutes. }
		\end{subfigure}
		\begin{subfigure}{.49\textwidth}
			\centering
			\includegraphics[width=8cm]{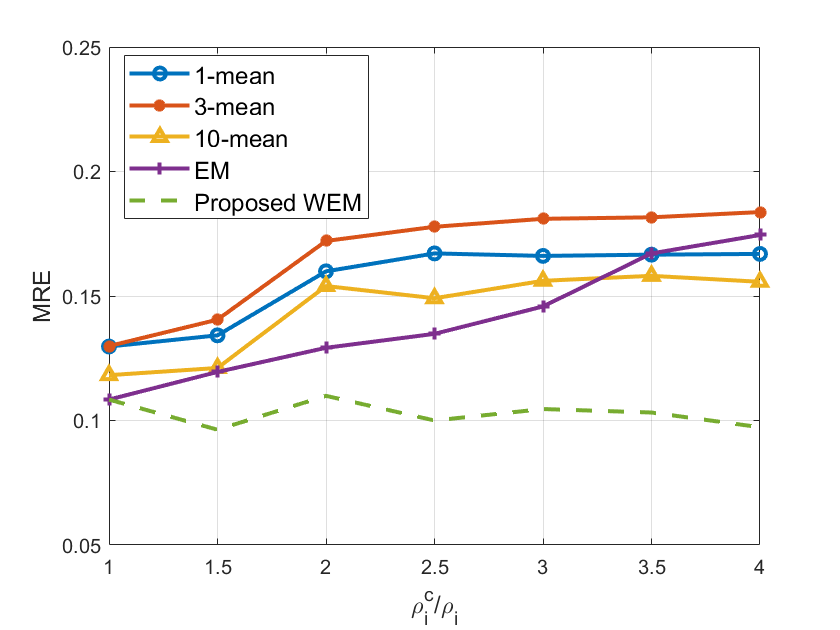}\vspace{0 cm}
			\caption{\small\label{sim:mre}{  MRE of the WEM approach and two baselines.} }  
		\end{subfigure}
		\vspace{-0.3 cm} 
		\caption{\small{\label{sim:learng}{  Statistical results and prediction errors in the learning stage. }
			} 
		}
	\end{center}
	\vspace{-1.2 cm}
\end{figure*}


Fig. \ref{sim:transPerUE2min} shows that over $70\%$ UEs receive, on average,  one packet within every two minutes. 
Thus, the transmission record  $\mathcal{S}_i$ that is collected during the learning stage ($\tau  = 2$ minutes) is a representative training datase. 
In this simulation,  the proposed WEM approach is applied to predict the data demand $d_i$, while the actual traffic demand $d_i^{\text{actual}}$ is calculated by summing up the real transmission amount within $\mathcal{A}_i^c$. 
Here, the mean relative error (MRE) is the metric to evaluate the prediction performance, where $\delta_{\text{MRE}} =  \mathbb{E}_{i,t}[\frac{|d_i - d_i^{\text{actual}}| } {d_i^{\text{actual}}}] $.   
Meanwhile, we introduce the  EM and $k$-mean methods as baselines. 
First, the EM method has been used in Section \ref{UEdistribution} for modeling the UE distribution $f_i(\boldsymbol{y})$. 
Here, to predict the traffic demand using the EM method, we have $ d_i^{\text{EM}} = T \cdot \mathbb{E}_n(s_n) \cdot  \int_{\boldsymbol{y} \in \mathcal{A}^{c}_i} f_i(\boldsymbol{y}) \dif \boldsymbol{y}$,  where $\mathbb{E}_n(s_n)=\frac{ \sum_n s_n \Delta t}{\tau}$ is the time-average data rate of all UEs, and $\int_{\boldsymbol{y} \in \mathcal{A}^{c}_i} f_i(\boldsymbol{y}) \dif \boldsymbol{y}$ is the percentage of UEs within the hotspot area.   
{\color{black} Note that, this is a commonly-used approach to estimate cellular data demand using the UE distribution and the average rate requirement per UE in the cellular network \cite{ mozaffari2018tutorial}.   } 
The $k$-mean method predicts the traffic density by averaging the local traffic density from $k$ closest neighbors.   

Fig. \ref{sim:mre} shows the prediction MRE of the WEM, EM, and $k$-mean methods, where $k = 1$, $3$ and $10$,  as the average data demand $\rho_i^c$ of the hotspot UEs increases.   
Note that, $\rho_i^c = \frac{1 }{Q^c_i}\int_{\boldsymbol{y} \in \mathcal{A}^c_i } S_i(\boldsymbol{y}) \dif \boldsymbol{y}$ is the average data rate per UE within the hotspot area, and  $\rho_i = \frac{1}{Q_i}\int_{\boldsymbol{y} \in \mathcal{A}_i } S_i(\boldsymbol{y}) \dif \boldsymbol{y} $ is the average rate demand of all UEs within the cellular network. 
When $\frac{\rho^c_i}{ \rho_i}=1$, each hotspot UE will have the same data demand as the other UEs. 
In this case, the WEM and EM approaches yield a similar prediction accuracy with an MRE of $11\%$, and the prediction errors of $k$-mean methods are between $12\%$ and  $12.5\%$. 
Note that, a prediction error of $11\%$ yields lower than $0.1$ W of deviation on the value of $p_{ij}(\boldsymbol{x}_{ij}^{*},\rho_i)$. 
Clearly, this is a very small value compared to the hovering  and transmit powers of a typical UAV.    
When the traffic load within different regions of the cellular network becomes more uneven, the prediction error of WEM remains the same, while the errors of the EM and $k$-mean methods gradually increase above $15.5\%$. 
Clearly, for $\frac{\rho^c_i}{ \rho_i} > 1$, the proposed WEM approach outperforms all other baselines.

In the WEM approach, the traffic density $\bar{{S}_i}(\boldsymbol{y})$ of each location $\boldsymbol{y}$ is considered when optimizing the prediction parameters.  Therefore, the spatial feature of downlink transmissions can be accurately captured, and the performance of WEM does not decrease when the traffic load in the cellular system becomes uneven.  
However, the EM model only considers the location information, but ignores the downlink rate of each transmission. Therefore, when the traffic demand shows distinct patterns in different regions, the EM method fails to capture the spatial diversity, and its prediction error increases significantly.  
Given that the $k$-mean method predicts the cellular traffic by averaging data from $k$ closest neighbors,  it captures the traffic spatial difference from local information.  
However, as the cellular traffic becomes more uneven, the local information is more sensitive to the noise, and thus, the prediction errors of $k$-mean methods increase, as $\rho_i^c$  increases.  
By comparing different $k$-mean algorithms, we find that  $3$-mean achieves the worst performance, because information from three neighbors is not sufficient to cancel out the noise.   
The simulation results also show that  $10$-mean yields the best performance among all $k$-mean methods. 

\subsection{The Impact of the UAV type on the utilities }
 
\begin{figure*}[!t]
	\begin{center}
		\vspace{-1cm}
		\begin{subfigure}{.49\textwidth}
			\centering
			\includegraphics[width=9cm]{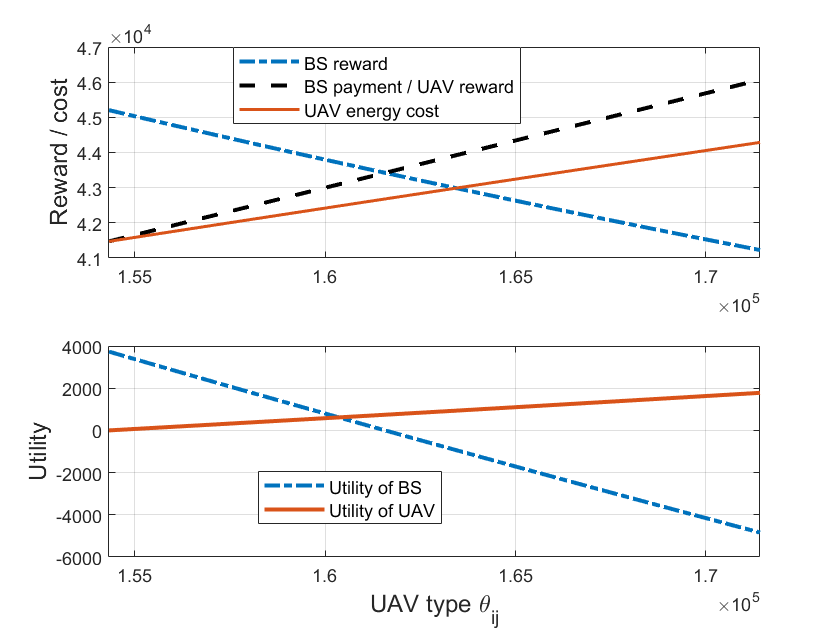}\vspace{0 cm}
			\caption{\label{sim:threeparts}{  Costs, rewards and overall utilities of the associated BS and UAV, given different UAV types. }}
		\end{subfigure}
		\begin{subfigure}{.49\textwidth}
			\centering
			\includegraphics[width=9cm]{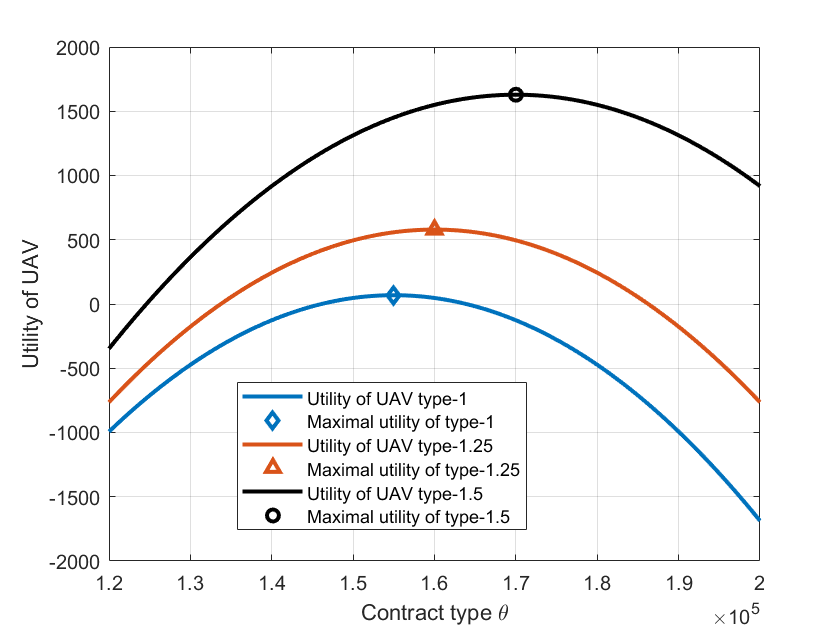}\vspace{-0.1 cm}
			\caption{\label{sim:IC} {  Utilities of UAVs, given different contract types. }} 
		\end{subfigure} 
		\vspace{-0.2 cm}
		\caption{\small{ \label{utility_type} {  As the UAV type increases, the transmit power and the unit payment both increase. However, the overall utilities of the associated BS and UAV will decrease. 	}}	 
		}
	\end{center}
	\vspace{-1cm}
\end{figure*}

In this section, we investigate the impact of the UAV type on the utilities. 
The contract is designed based on Lemma \ref{lemma1} by a BS with ID 7939, using the data from time $42$ in \cite{cityCellularTrafficMap}.  
Fig. \ref{sim:threeparts} shows the relationship between the UAV type and the reward, cost, as well as the overall utilities of the requesting BS  and the deployed UAV, respectively.  
First, as the UAV type $\theta_{ij}$ becomes larger,  the BS's reward $\beta B_{ij}(p_j)$ from the downlink UEs will decrease. 
Although the transmit power $p_j(\theta_{ij})(\theta_{ij})$ becomes higher given a larger $\theta_{ij}$,
a  UAV with a higher type  must travel for a longer time $t_{ij}$ before its service. Thus, the downlink data transmission of the UAV becomes lower for a larger  $\theta_{ij}$.
Meanwhile, based on (\ref{aerialCapacity}), we have $\frac{\dif B_{ij}(\theta_{ij})}{ \dif \theta_{ij}} <0$. 
Therefore, 
a higher UAV type $\theta_{ij}$ leads to a lower BS's reward.   
Furthermore, a higher UAV type increases the payment $u_i(\theta_{ij})d_i$ from BS $i$ to UAV $j$, and, thus, the utility of BS $i$ will be lower. 
For the deployed UAV $j$, a larger $\theta_{ij}$ results in a higher reward $u_i(\theta_{ij})d_i$ from BS $i$, and the increase of the UAV's reward is faster than the energy cost. Therefore,  as shown in Fig. \ref{sim:threeparts}, 
the utility of the deployed UAV will increase, as its type $\theta_{ij}$ becomes larger.    
Meanwhile, Fig. \ref{sim:threeparts} shows that the UAV's utility is always non-negative. Therefore,  the IR condition holds in the designed contract. 
Fig. \ref{sim:IC} investigates the impact of the contract type on the UAV's utility. 
The utilities of three UAVs, where their actual types are $1 \times 10^5$ (type-1), $1.25 \times 10^5$ (type-1.25), and  $1.5 \times 10^5$ (type-1.5),  
is given, 
when they accept different kinds of contracts from BS 7939. 
As shown in Fig. \ref{sim:IC},  the maximum utility of each UAV is achieved when the accepted contract is of its own type.  
Thus,  simulation results show that the IC condition holds in the designed contract set. 

An interesting observation on the utility function is that the prediction error of $d_i$ does not cause small fluctuations on the utility value of the BS or the employed UAV.  
Given the transmit power as $p_j =  \gamma_i \theta_{ij}^2/2 = \frac{mT^2}{4\alpha(T-t_{ij})^2}$ and the total payment from BS $i$ to UAV $j$ as $u_id_i = \gamma_i \theta_{ij} d_i = \frac{mT^2}{2(T-t_{ij})}$,  $d_i$ no longer appears in the utility formulas, and, thus, an inaccuracy in $d_i$ will not impact the utility functions in (\ref{U_BS}) and (\ref{U_UAV}).  
The main effect of  $d_i$ in the predictive UAV deployment is to determine the minimum required transmit power  $p_{ij}(\boldsymbol{x}_{ij}^{*},\rho^c_i)$.  
If the predicted demand $d_i$ is much lower than the real data demand, then $p_{ij}(\boldsymbol{x}_{ij}^{*},\rho^c_i)$ will be smaller. In consequence, some UAVs without enough energy may be inappropriately  considered to be a qualified choice, and might be employed. 
On the other hand, if $d_i$ is much higher than the actual demand, some qualified UAVs with enough power may be excluded from the candidate set $\mathcal{J}_i$. 
Both cases can lead to a suboptimal solution to (\ref{problem_BS}).  
However, as long as the error on $d_i$ causes no change to the association result,  the utilities of the BS and UAV will always be accurate.     
Based on this observation, the proposed  approach is highly robust to  prediction errors.  

\subsection{Evaluation of the predictive UAV deployment} 
 
In this section, we evaluate the  performance of the proposed UAV deployment method with four metrics, which are the  downlink capacity,  energy consumption,  service delay of the employed UAVs, and the utilities of the BS and  UAV operators. 
Meanwhile, for comparison purposes,  an event-driven deployment of the closest UAV and an event-driven deployment of a UAV with the maximal on-board energy are introduced as two baselines.  
In both baseline approaches, the target UAV is requested by the overloaded BS and deployed, after the downlink congestion occurs, without the prediction on traffic demand. 
The optimal location of the deployed UAV  is  determined after the UAV arrives at the service area, so as to maximize its downlink transmission rate \cite{mozaffari2016optimal}.  
Meanwhile, in both baseline approaches,  there is no contract design to determine the cost and payment between the BS to its employed UAV. 
Instead, the employed UAV $j$  provides the downlink service to the best of its power ability, where  $p_j =  \min \{p_{ij}(\boldsymbol{x}_{ij}^{*},\rho_i),  p_{ij}^{\text{max}}, p_{\text{max}} \}$, 
and  the unit payment $u_i$ from BS $i$ to the employed UAV is a fixed price $\beta$, which equals to the unit payment from the UEs to the BS per bit of data service.

\begin{figure*}[!t]
	\begin{center}
		\vspace{-1cm}
		\begin{subfigure}{.47\textwidth}
			\centering
			\includegraphics[width=7.3cm]{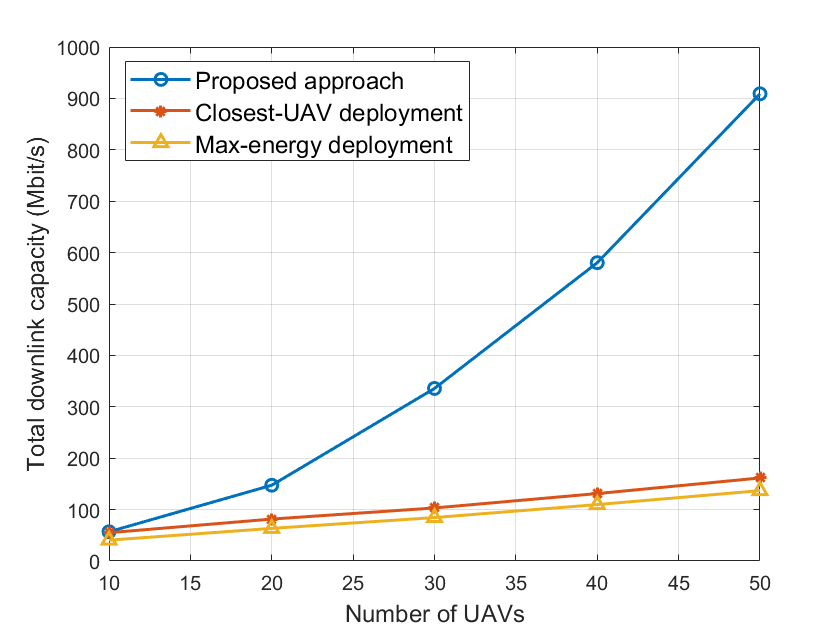}\vspace{-0.2 cm}
			\caption{\label{capacity}{ Total downlink capacity of the employed UAVs }}   
		\end{subfigure}
		\begin{subfigure}{.47\textwidth}
			\centering
			\includegraphics[width=7.3cm]{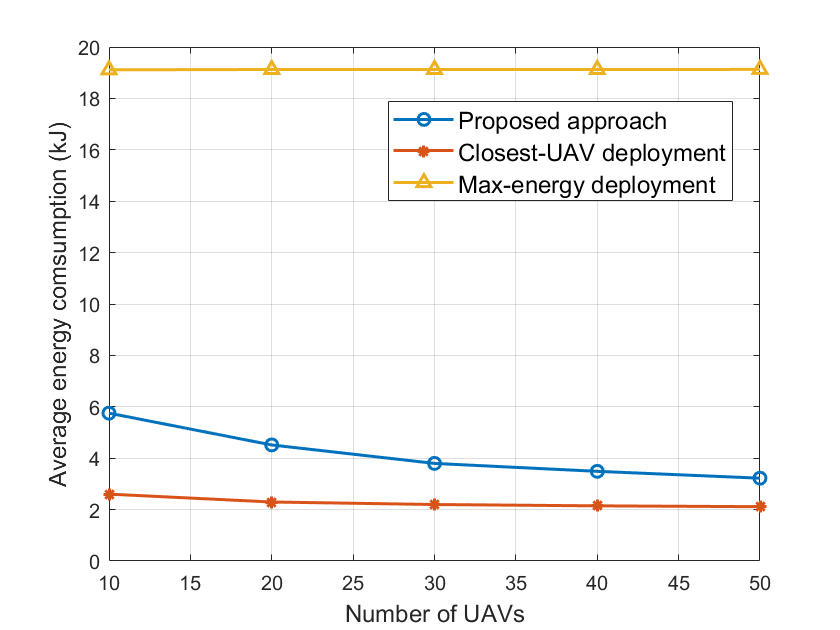}\vspace{-0.2 cm}
			\caption{\label{energy}{ Average energy consumption  per UAV }}  
		\end{subfigure}
		\begin{subfigure}{.47\textwidth}
			\centering
			\includegraphics[width=7.3cm]{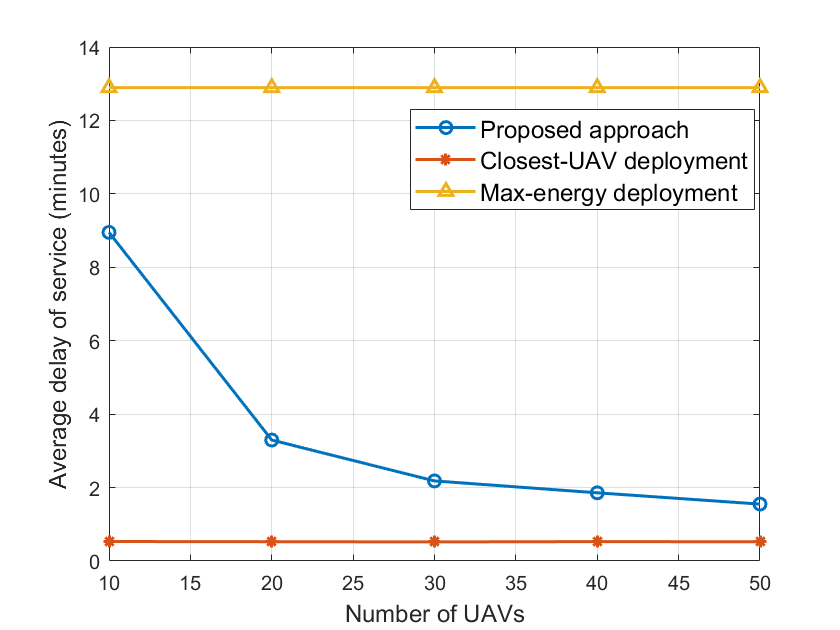}\vspace{-0.2 cm}
			\caption{\label{delay}{  Average service delay per UAV  }}    
		\end{subfigure}
		\vspace{-0.2cm}
		\caption{\small{ \label{metrics}  { Total downlink capacity, average energy consumption and average service delay of the UAV downlink service for the proposed predicted UAV deployment and two baselines. } } 
		}
	\end{center}
	\vspace{-1.2cm}
\end{figure*}

In Fig. \ref{metrics}, we compare the performance of the proposed predictive UAV deployment with two baselines, in terms of the total  downlink capacity, average energy consumption, and average service delay of the employed UAVs. 
First,  in Fig. \ref{capacity}, as the number of UAVs within the cellular network increases, the total downlink capacity that the employed UAVs provide to the downlink UEs  increases in all three schemes.  
For a larger number of UAVs, the average movement distance between each overloaded BS and  its employed UAV will decrease. 
Therefore, less energy is consumed during the mobility stage, and more power can be reserved for the downlink transmission service. 
In consequence, the downlink capacity of all three methods increases. 
However, in the closest-UAV approach,  without the data demand prediction, the deployed closest UAV may not have enough on-board energy to satisfy the downlink data demand. Therefore,  the  downlink capacity in the closest-UAV baseline is lower than the proposed approach. 
In the max-energy deployment, the distance between the employed UAV and the service area is usually larger, compared to two other methods.  Although the employed UAV has the largest amount of available onboard energy, due to a longer travel distance, most of the onboard energy will be consumed on  mobility, and the transmit power may be insufficient. 
Thus,   the max-energy deployment yields the lowest capacity performance among all three schemes.  
Moreover, the proposed approach improves the downlink capacity by over four-fold and five-fold, compared to the closest-UAV and the max-energy baselines, respectively.

Fig. \ref{energy} and Fig. \ref{delay} show the average energy consumption  and  service delay of each employed UAV, respectively, as the number of available UAVs in the network increases. 
First, we can see that the closest-UAV scheme yields the least energy cost and service delay, due to its shortest movement distance. 
In the proposed approach, the energy consumption and movement duration are relatively higher, because the selection criteria balances between the distance of the UAV (which determines the movement energy) and the availability of sufficient on-board energy to meet the predicted data demand. 
Meanwhile, the max-energy deployment results in the highest energy and time cost, due to the largest travel distance during the mobility stage.  
Next, for a higher number of UAVs, the energy consumption and service delay of the proposed method both drop, while the performance of the baselines remains nearly constant. 
In particular, as the number of UAVs increases, the performance of the proposed approach improves exponentially, and the gap between the proposed approach and the closest-UAV scheme becomes much smaller.  
In the proposed method, having more UAVs reduces the average distance between any employed UAV and its service point, and, hence, decreases the energy and time cost. 
However, in the two baselines, the number of available UAVs does not effect the travel distance during mobility. Thus, the energy consumption and service delay of two baselines remain nearly constant with the increase in the number of UAVs.

In Fig. \ref{utilities}, we compared the utilities of the BS and UAV operators in three schemes.  
First, in Fig. \ref{utilitybs}, for a larger number of UAVs, the average utility per BS increases in all three schemes,  and the proposed approach yields the highest utility.   
In the proposed method, by having more UAVs, the average distance between an employed UAV to its service becomes smaller, and, thus, the type $\theta_{ij}$ of the employed UAV $j$ with respect to the requesting BS $i$ decreases, which yields a higher utility of  BS $i$. 
For the closest-UAV and max-energy schemes, since the employed UAV cannot always satisfy the data demand of its downlink UEs, the utilities of each BS for both baselines are lower, compared the proposed method. 

In Fig. \ref{utilityuav}, we can see that, as the number of UAVs increases,  the total utility of the employed UAVs becomes higher in the proposed approach, while the UAVs' utilities resulting from both baseline schemes are much lower than the proposed method.   
As shown in Figs. \ref{capacity} and \ref{energy}, by having more UAVs, the average energy cost per UAV resulting from the proposed approach will decrease, while the downlink transmission capacity of the UAV networks increases, which yields a higher income.   
As a result, the overall utility of the UAV operator in the proposed method will become higher for a larger  number of  UAVs.    
For the closest-UAV scheme,  its lower  energy consumption and shorter service delay yield  a smaller deployment cost, compared with the proposed method. 
However, the lower downlink capacity results in less payment from the BS. Thus, the total utility of the UAV operator in  the closest-UAV scheme is less than the proposed method. 
Moreover,  based on Figs. \ref{capacity}, \ref{energy}, and \ref{delay}, we can see that the max-energy scheme yields the lowest transmission rate, the highest energy cost, and the longest service delay. Therefore, the utility of the UAV operators in  the max-energy scheme is the lowest among all three methods. 
In consequence, based on Fig. 4 and Fig. 5, we can conclude that  the proposed method enables an efficient UAV deployment to alleviate communication congestion in the cellular networks, and shows a significant advantage on the economical revenues of both the BS and UAV operators, compared with two baseline, event-driven approaches.

\begin{figure*}[!t]
	\begin{center}
		\vspace{-1cm}
		\begin{subfigure}{.49\textwidth}
			\centering
			\includegraphics[width=8cm]{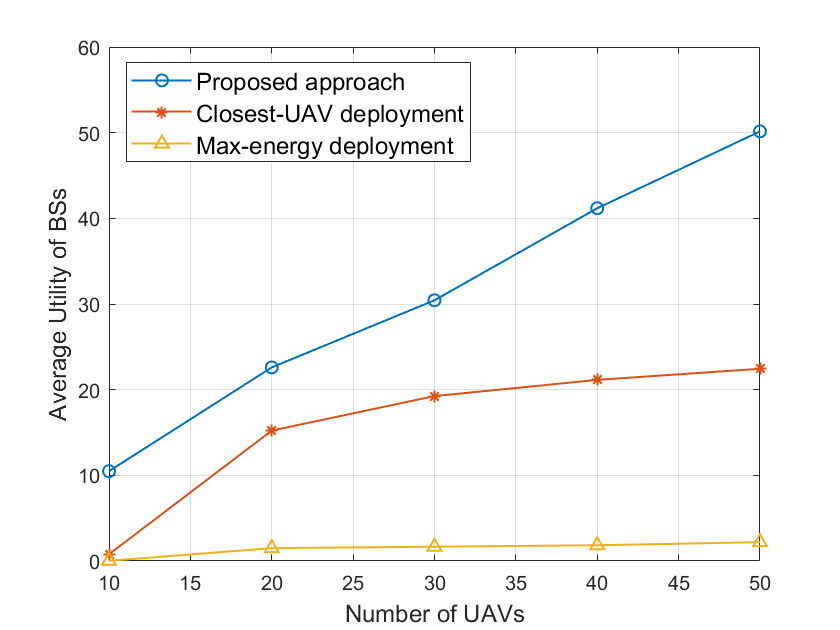}\vspace{0 cm}
			\caption{\label{utilitybs}   Average utility of BSs}  
		\end{subfigure}
		\begin{subfigure}{.49\textwidth}
			\centering
			\includegraphics[width=8cm]{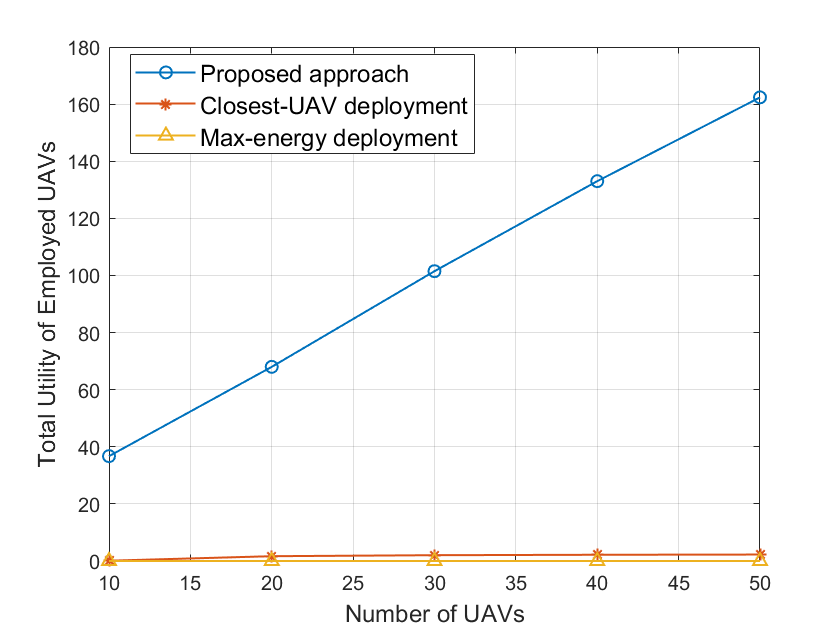}\vspace{0 cm}
			\caption{\small\label{utilityuav}{   Total utility of the UAV operators} } 
		\end{subfigure}
		\vspace{-0.3 cm} 
		\caption{\small{\label{utilities}{   Utility of the BS and UAV operators for the proposed predictive UAV deployment and two baselines. } 
			} 
		}
	\end{center}
	\vspace{-1.2 cm}
\end{figure*}

\section{Conclusion}\label{conclusion} 
In this paper, we have proposed a novel approach for predictive deployment of UAVs to complement the ground cellular system in face of the hotspot events. 
In particular,  four inter-related and sequential stages have been proposed to enable the ground BS to optimally employ a UAV to offload the excess traffic. 
First,  a novel framework, based on the EM and WEM methods, has been proposed to estimate the UE distribution and the downlink traffic demand.
Next, to guarantee a truthful information exchange between the BS and UAV operators, a traffic offload contract have been developed, and the sufficient and necessary conditions for having a feasible contract have been analytically derived. 
Then, an optimization problem have been formulated to deploy the optimal UAV onto the hotspot area in a way that the utility of each overloaded ground BS is maximized.     
Simulation results show that the proposed WEM approach yields a prediction error around $10\%$, and compared with the EM and $k$-mean schemes, the WEM algorithm yields a higher prediction accuracy, particularly when the traffic load in the cellular system becomes spatially uneven.    
Furthermore, compared with two event-driven schemes based on the closest-distance and maximal-energy metrics, the proposed predictive deployment approach enables UAV operators to provide efficient downlink service for hotspot users, and significantly improves the revenues of both the BS and UAV networks.

\appendices
\section{Proof of Proposition 1}\label{AppPro1}
We first use contradiction to prove the proposition that if $\theta_{ij} > \theta_{ij}^{'}$, then $u_i(\theta_{ij}) \ge u_i(\theta_{ij}^{'})$. 
Suppose that there exists $u_i(\theta_{ij})<u_i(\theta_{ij}^{'})$, but $\theta_{ij}>\theta_{ij}^{'}$. Then, we have
	\begin{align}\label{AppPro2Ine0}
	\theta_{ij}u_i(\theta_{ij}^{'}) + \theta_{ij}^{'} u_i(\theta_{ij}) > \theta_{ij}u_i(\theta_{ij}) + \theta_{ij}^{'} u_i(\theta_{ij}^{'}).
	\end{align}
	On the other hand, from  IC condition, we have
	\begin{align}\label{AppPro2Ine1}
	\theta_{ij}u_i(\theta_{ij}) - p_{j}(\theta_{ij}) \ge \theta_{ij}u_i(\theta_{ij}^{'}) - p_{j}(\theta_{ij}^{'}),~~~~ \text{   } ~~~~\theta_{ij}^{'}u_i(\theta_{ij}^{'}) - p_{j}(\theta_{ij}^{'}) \ge \theta_{ij}^{'} u_i(\theta_{ij}) - p_{j}(\theta_{ij}). 
	\end{align} 
	By adding the inequations in (\ref{AppPro2Ine1}), we have $\theta_{ij}u_i(\theta_{ij})  + \theta_{ij}^{'}u_i(\theta_{ij}^{'}) \ge \theta_{ij}u_i(\theta_{ij}^{'}) + \theta_{ij}^{'} u_i(\theta_{ij})$, 
	which contradicts to (\ref{AppPro2Ine0}). This completes the first part of the proof. 
	
	Next, we prove that if $u_i(\theta_{ij}) \ge u_i(\theta_{ij}^{'})$,  $p_{j}(\theta_{ij}) \ge p_{j}(\theta_{ij}^{'})$. From the IC condition, we have
	$\theta_{ij}^{'} u_i(\theta_{ij}^{'}) - p_{j}(\theta_{ij}^{'}) \ge \theta_{ij}^{'} u_i(\theta_{ij}) -p_{j}(\theta_{ij})$, 
	i.e. $p_{j}(\theta_{ij}) - p_{j}(\theta_{ij}^{'}) \ge \theta_{ij}^{'} \left( u_i(\theta_{ij}) - u_i(\theta_{ij}^{'}) \right)$. 
	Since $u_i(\theta_{ij}) > u_i(\theta_{ij}^{'})$, we conclude 
	$p_{j}(\theta_{ij}) - p_{j}(\theta_{ij}^{'}) \ge \theta_{ij}^{'} \left( u_i(\theta_{ij}) - u_i(\theta_{ij}^{'}) \right) \ge 0$, 
	and thus $p_{j}(\theta_{ij}) \ge p_{j}(\theta_{ij}^{'})$. This completes the proof. 
 
\section{Proof of Theorem \ref{contracttheorem}} \label{AppTheo1}
For notation  simplicity, in this section, we 
denote  $u_i$, $p_j$, $\theta_{ij}$,  $M_{ij}$ 
as  $u$, $P$, $\theta$, $M$ respectively.  
\subsection{Proof for necessary conditions} 
Given the IR and IC conditions, we prove Theorem \ref{contracttheorem} in this section. 
First, as shown in Proposition \ref{pro}, for any $\theta, \theta^{'} \in \Theta_i$, once $\theta > \theta^{'}$, then $u(\theta) \ge u(\theta^{'})$ and $P(\theta) \ge P(\theta^{'})$. Therefore, condition (a) of Theorem \ref{contracttheorem} is proved by Proposition 1.  
Second,  condition (b) of Theorem \ref{contracttheorem} is supported by the IR condition, where $R_{ij}(\theta) \ge 0$ for all $\theta$ in $\Theta_i$, which naturally includes $\theta^{\text{min}}$.  
Next, we  prove condition (c).  Let $\Delta \theta = \theta^{'} - \theta $. 
According to the IC condition, for any $\Delta \theta \in [\theta^{\text{min}} - \theta^{\text{max}},0) \cup (0, \theta^{\text{max}} - \theta^{\text{min}}  ]$, we have: 
$\theta \cdot u(\theta) - P(\theta)  \ge \theta \cdot u(\theta + \Delta \theta) - P(\theta + \Delta \theta)$,  
i.e., $\theta \cdot [u(\theta) - u(\theta + \Delta \theta)]  \ge   P(\theta) - P(\theta + \Delta \theta)$. 
If $\Delta \theta >0$, then according to Proposition 1, $ u(\theta + \Delta \theta) \ge u(\theta) $ and $ P(\theta + \Delta \theta) \ge P(\theta)$. 
Here, we exclude the situation where $ u(\theta + \Delta \theta) = u(\theta) $ and $ P(\theta + \Delta \theta) = P(\theta)$ in the following discussion of this proof, because condition (c) naturally holds in this case.  
Therefore, for any $\Delta \theta \in  (0, \theta^{\text{max}} - \theta^{\text{min}}]$, we have  
\begin{align}\label{you}
	\theta \le \frac{P(\theta + \Delta \theta) -  P(\theta)}{u(\theta + \Delta \theta) - u(\theta)}.
\end{align}
If $\Delta \theta <0$,  then $ u(\theta + \Delta \theta) < u(\theta) $ and $ P(\theta + \Delta \theta) < P(\theta)$. Thus, for any $\Delta \theta \in [\theta^{\text{min}} - \theta^{\text{max}},0)$, 
\begin{align}\label{zuo}
\theta \ge \frac{P(\theta + \Delta \theta) -  P(\theta)}{u(\theta + \Delta \theta) - u(\theta)}.
\end{align} 
Combing (\ref{you}) and (\ref{zuo}), we have 
 $\frac{\dif P }{ \dif \theta} / \frac{\dif u }{\dif \theta }  =   	\lim_{\Delta \theta \to 0} \frac{P(\theta + \Delta \theta) - P(\theta)}{ u(\theta + \Delta \theta) - u(\theta) } = \theta$,  
which proves condition (c) of Theorem \ref{contracttheorem}.  

\subsection{Proof for sufficient conditions} 
From Theorem \ref{contracttheorem}, we will prove the IR and IC conditions in this section.  
First, we prove the IR condition. According to condition (b) of Theorem \ref{contracttheorem}, $\theta^{\text{min}}$ satisfies the IR condition. Then, we prove that for any $\theta \in (\theta^{\text{min}}, \theta^{\text{max}}]$, the IR condition holds. 
From condition (c) of Theorem \ref{contracttheorem}, we have the following inequalities, $\frac{P(\theta) -P(\theta^{\text{min}})}{u(\theta) -u(\theta^{\text{min}})}  \le \theta$,  i.e.,
\begin{align}\label{pmin1}
	P(\theta^{\text{min}}) \ge P(\theta) - \theta \cdot [ u(\theta) - u(\theta^{\text{min}}) ]. 
\end{align}
From condition (b), we have
\begin{align}\label{pmin2}
	\theta^{\text{min}} \cdot u(\theta^{\text{min}}) - P(\theta^{\text{min}}) - M \ge 0.
\end{align}
By combing (\ref{pmin1}) and (\ref{pmin2}), we have $\theta \cdot u(\theta) - P(\theta) - M\ge (\theta - \theta^{\text{min}}) \cdot u(\theta^{\text{min}}) \ge 0$. 
Thus, for any $\theta \in \Theta_i$, the IR condition holds. 

In the end, we prove the IC condition.  
Let $h = \theta \cdot u(\theta) - P(\theta) - M - [ \theta \cdot u(\theta^{'}) - P(\theta^{'}) - M ]$. And we prove that $h \ge 0$. 
From condition (c), we have, if $\theta^{'} > \theta$, then
$     \frac{P(\theta^{'}) -P(\theta )}{u(\theta^{'}) -u(\theta)} \ge \min \{ \theta, \theta^{'} \} = \theta$. 
i.e., $	P(\theta^{'}) -P(\theta ) \ge \theta \cdot [u(\theta^{'}) -u(\theta)]$. 
Therefore, $h = \theta \cdot [u(\theta) -  u(\theta^{'})] +   P(\theta^{'}) - P(\theta) \ge 0$. 
On the other hand, if  $\theta^{'} < \theta$, then
$\frac{P(\theta) -P(\theta^{'} )}{u(\theta) -u(\theta^{'})} \le \max \{ \theta, \theta^{'} \} = \theta$. 
i.e., $P(\theta) -P(\theta^{'} ) \le \theta \cdot [u(\theta) -u(\theta^{'})]$. 
Therefore, $h \ge 0$. Consequently, the IC condition holds.  
\ifCLASSOPTIONcaptionsoff
  \newpage
\fi

\bibliographystyle{IEEEtran}
\bibliography{references}

\end{document}